%% file: main.tex
\documentclass[a4paper,10pt]{article}
\usepackage{a4wide}
\input{macros}

\begin{document}

\title{Effectful Applicative Bisimilarity:\\ 
Monads, Relators, and Howe's Method\\ (Long Version)}

\date{}
\author{Ugo Dal Lago \and Francesco Gavazzo \and Paul Blain Levy}

\maketitle

\begin{abstract}
  We study Abramsky's applicative bisimilarity abstractly, in the
  context of call-by-value $\lambda$-calculi with algebraic
  effects. We first of all endow a computational $\lambda$-calculus
  with a monadic operational semantics. We then show how the theory of
  relators provides precisely what is needed to generalise applicative
  bisimilarity to such a calculus, and to single out those monads and
  relators for which applicative bisimilarity is a congruence, thus a
  sound methodology for program equivalence. This is done by studying
  Howe's method in the abstract.
\end{abstract}

\section{Introduction}
Program equivalence is one of the central notions in the theory of
programming languages, and giving satisfactory definitions and
methodologies for it is a challenging problem, for example when
dealing with higher-order languages.  The problem has been approached,
since the birth of the discipline, in many different ways. One can
define program equivalence through denotational semantics, thus
relying on a model and stipulating two programs to be equivalent 
if and only if they are interpreted by the same denotation. 
If the calculus at hand is equipped with a
notion of \emph{observation}, typically given through some forms of
operational semantics, one could proceed following the route traced by
Morris, and define programs to be \emph{contextual} equivalent when they
\emph{behave} the same \emph{in every} context.

Both these approaches have their drawbacks, the first one relying on
the existence of a (not too coarse) denotational model, the latter
quantifying over all contexts, and thus making concrete proofs of
equivalence hard. Handier methodologies for proving programs equivalent
have been introduced along the years based on logical relations and
applicative bisimilarity. Logical relations were originally devised
for typed, normalising languages, but later generalised to more
expressive formalisms, e.g., through
step-indexing~\cite{AppelMcAllester/TOPLAS/2001} and
biorthogonality~\cite{BentonKennedyBeringerHofmann/PPDP/2009}. Starting
from Abramsky's pioneering work on applicative
bisimilarity~\cite{Abramsky/RTFP/1990}, coinduction has also been
proved to be a useful methodology for program equivalence, and has
been applied to a variety of calculi and language features.

The scenario just described also holds when the underlying calculus
is not pure, but effectful. There have been many attempts to study
effectful
$\lambda$-calculi~\cite{PlotkinPower/FOSSACS/01,Moggi/LICS/89} by way
of denotational semantics \cite{Jones/PhDThesis,deLiguoroPiperno/IC/1995,DanosHarmer/TOCL/2002},
logical relations~\cite{BizjakBirkedal/FOSSACS/2015}, and applicative
bisimilarity~\cite{Lassen/PhDThesis,DalLagoSangiorgiAlberti/POPL/2014,CrubilleDalLago/ESOP/2014}.
But while the denotational and logical relation semantics of effectful
calculi have been studied in the
abstract~\cite{GoubaultLasotaNowak/MSCS/2008,JohannSimpsonVoigtlander/LICS/2010},
the same cannot be said about applicative bisimilarity and related
coinductive techniques.  There is a growing body of literature on
applicative bisimilarity for calculi with, e.g.,
nondeterministic~\cite{Lassen/PhDThesis}, and
probabilistic effects~\cite{DalLagoSangiorgiAlberti/POPL/2014}, but
each notion of an effect has been studied independently, often getting different
results. Distinct proofs of congruence for applicative bisimilarity,
even if done through a common methodology, namely the so-called Howe's
method~\cite{Howe/IC/1996}, do not at all have the same difficulty
in each of the cases cited above. As an example, the proof of
the so-called Key Lemma relies on duality results from linear programming
\cite{Schrijver/Book/1986} when done for probabilistic effects,
contrarily to the apparently similar case of nondeterministic effects,
whose logical complexity is comparable to that for the plain, deterministic
$\lambda$-calculus \cite{Ong/LICS/1993,Lassen/PhDThesis}. Finally, as
the third author observed in his work with Koutavas and
Sumii~\cite{KoutavasLevySumii/ENTCS/2011}, applicative bisimilarity is fragile to
the presence of certain effects, like local states or dynamically 
created exceptions: in these cases, a sort of information hiding is possible
which makes applicative bisimilarity simply too weak, and thus unsound
for contextual equivalence.

The observations above naturally lead to some questions. Is there any
way to factor out the common part of the congruence proof for 
applicative bisimilarity in the cases above? 
Where do the limits on the correctness of applicative
bisimilarity lie, in presence of effects? The authors strongly believe
that the field of coinductive techniques for higher-order program
equivalence should be better understood \emph{in the abstract}, this
way providing some answers to the questions above, given that generic
accounts for effectful $\lambda$-calculi abound in the
literature~\cite{Moggi/LICS/89,PlotkinPower/FOSSACS/01}.

This paper represents a first step towards answering the questions
above. We first of all introduce a computational $\lambda$-calculus in
which general algebraic effects can be represented, and give a monadic
operational semantics for it, showing how the latter coincides with
the expected one in many distinct concrete examples. We then show how
applicative bisimilarity can be defined for any instance of such a
monadic $\lambda$-calculus, based on the notion of a relator, which
allows to account for the possible ways a relation on a set $X$ can be
turned into one for $\monad X$, where $\monad$ is a monad.  We then
single out a set of axioms for monads and relators which allow us to
follow Howe's proof of congruence for applicative bisimilarity
\emph{in the abstract}. Noticeably, these axioms are satisfied in all
the example algebraic effects we consider. The proof of it allows us
to understand the deep reasons \emph{why}, say, different instances of
Howe's method in the literature seem to have different complexities.

\section{On Coinduction and Effectful $\lambda$-Calculi}\label{sect:informal}
In this section, we illustrate how coinduction can be useful when
proving the equivalence of programs written in higher-order effectful
calculi.

\begin{figure}
  \begin{center}
  \fbox{
    \begin{minipage}{.42\textwidth}
      \begin{center}
        \begin{align*}
          W&\rightarrow V\oplus\mathit{COMP}(V,W)\\
          Z&\rightarrow T\,\underline{1}\\
          T\,\underline{n}&\rightarrow (R\,\underline{n})\oplus(T\,\underline{n+1})\\
          R\,\underline{0}&\rightarrow\lambda x.x\\
          R\,\underline{n+1}&\rightarrow\mathit{COMP}(R\,\underline{n},V)\\
        \end{align*}
      \end{center}
  \end{minipage}}
  \end{center}
  \caption{Two Probabilistic Programs.}\label{fig:examplefirst}
\end{figure}
Let us start with a simple example of two supposedly equivalent
probabilistic functional programs, $W$ and $Z$, given in Figure~\ref{fig:examplefirst}.
(The expression $\mathit{COMP}(M,N)$ stands for the term $\lambda y.M(Ny)$, and 
$\oplus$ is a binary operation for fair probabilistic choice.)
Both $W$ and $Z$ behave like the $n$-th composition of a function $V$ with itself
with probability $\frac{1}{2^n}$, for every $n$. But how could we
even \emph{define} the equivalence of such effectful programs?
A natural answer consists in following Morris \cite{Morris/PhDThesis}, and stipulate
that two programs are contextually equivalent if they behave the same when put
in any context, where the observable behaviour of a term can be
taken, e.g., as its probability of convergence. Proving two terms
to be contextually equivalent can be quite hard, given the universal
quantification over all contexts on which contextual equivalence is based.

Applicative bisimilarity is an alternative definition of program
equivalence, in which $\lambda$-terms are seen as computational
objects interacting with their environment by exposing their
behaviour, and by taking arguments as input. Applicative bisimilarity
has been generalised to effectful $\lambda$-calculi of various kinds,
and in particular to untyped probabilistic
$\lambda$-calculi~\cite{DalLagoSangiorgiAlberti/POPL/2014}, and it is
known to be not only a congruence (thus \emph{sound} for contextual
equivalence) but also \emph{fully abstract}, at least for
call-by-value evaluation~\cite{CrubilleDalLago/ESOP/2014}. Indeed,
applicative bisimilarity can be applied to the example terms in
Figure~\ref{fig:examplefirst}, which can this way be proved contextual
equivalent.

The proof of soundness of applicative bisimilarity in presence of
probabilistic effects is significantly more complicated than the
original one, although both can be done by following the so-called
Howe's method~\cite{Howe/IC/1996}. More specifically, the proof that the Howe extension
of similarity is a simulation relies on duality from linear programming 
(through the Max Flow Min Cut Theorem)
when done in presence of probabilistic effects, something that is
not required in the plain, deterministic setting, nor in presence of
\emph{nondeterministic} choice.

Modern functional programming languages, however, can be ``effectful''
in quite complex ways. As an example, programs might be allowed not
only to evolve probabilistically, but also to have an internal state, to
throw exceptions, or to perform some input-output operations.
Consider, as another simple example, the programs in
Figure~\ref{fig:examplesecond}, a variation on the programs from
Figure~\ref{fig:examplefirst} where we allow programs to additionally 
raise an exception $e$ by way of the $\raiseExcep{e}$ command. Intuitively,
$W^{\mathsf{raise}}$ and $Z^{\mathsf{raise}}$ behave like $W$ and $Z$, respectively,
but they both raise an exception with a certain probability.
\begin{figure}
  \begin{center}
    \fbox{
      \begin{minipage}{.47\textwidth}
        \begin{center}
          \begin{align*}
            W^{\mathsf{raise}}&\rightarrow (V\oplus\raiseExcep{e})\oplus\mathit{COMP}(V,W^{\mathsf{raise}})\\
            Z^{\mathsf{raise}}&\rightarrow T\,\underline{1}\\
            T\,\underline{n}&\rightarrow ((R\,\underline{n})\oplus \raiseExcep{e})\oplus(T\,\underline{n+1})\\
            R\,\underline{0}&\rightarrow\lambda x.x\\
            R\,\underline{n+1}&\rightarrow\mathit{COMP}(R\,\underline{n},V)\\
          \end{align*}
        \end{center}
    \end{minipage}}
  \end{center}
  \caption{Two Probabilistic Programs Throwing Exceptions.}\label{fig:examplesecond}
\end{figure}

While applicative similarity in presence of catchable exceptions is
well-known to be unsound~\cite{KoutavasLevySumii/ENTCS/2011}, the mere presence of the
$\raiseExcep{e}$ command does not seem to cause any significant
problem. The literature, however, does not offer any
result about whether \emph{combining} two or more notions of
computational effect for which bisimilarity is known to work well,
should be problematic or not. An \emph{abstract}
theory accounting for how congruence proofs can be carried out in
effectful calculi is simply lacking.

Even if staying within the scope of Howe's method, it seems that
each effect between those analysed in the literature is handled by way
of some \emph{ad-hoc} notion of bisimulation. As an example,
nondeterministic extensions of the $\lambda$-calculus can be dealt
with by looking at terms as a labelled transition system, while
probabilistic extensions of the $\lambda$-calculus require a different
definition akin to Larsen and Skou's probabilistic
bisimulation~\cite{DalLagoSangiorgiAlberti/POPL/2014}. What kind of
transition do we need when, e.g., dealing with the example from
Figure~\ref{fig:examplesecond}?  In other words, an abstract theory of
effectful applicative bisimilarity would be beneficial from a purely
definitional viewpoint, too.

What could come to the rescue here is the analysis of effects and
bisimulation which has been carried out in the field of
coalgebra~\cite{Rutten/TCS/2000}.  In particular, we here exploit the
theory of relators, also known as lax extensions~\cite{Barr/LMM/1970,Thijs/PhDThesis/1996}.
\section{Domains and Monads: Some Preliminaries}
In this section, we recall some basic definitions and results on
complete partial orders, categories, and monads. All will be central
in the rest of this paper. Due to space constraints, there is no hope to be
comprehensive. We refer to the many introductory textbooks on partial order 
theory \cite{DaveyPriestley/Book/1990} or category theory \cite{MacLane/Book/1971} 
for more details.
\subsection{Domains and Continuous $\signature$-algebras}

Here we recall some basic notions and results on domains that we will 
extensively use in this work. The main purpose of this section is to introduce 
the notation and terminology we will use in the rest of this paper. We address the 
reader to e.g. \cite{AbramskyJung/DomainTheory/1994} for a deeper treatment of the 
subject.

Recall that a poset is a set equipped with a reflexive, transitive
and antisymmetric relation.

\begin{definition}\label{omegaChain}
  Given a poset $\Cpoone = (\cpoone, \cpoleq_{\cpoone})$, an $\omega$-chain in
  $\cpoone$ is an infinite sequence $(x_n)_{n < \omega}$ of elements of
  $\cpoone$ such that $x_n \cpoleq_{\cpoone} x_{n+1}$, for any $n \geq 0$.
\end{definition}

\begin{definition}\label{omegaCPPO}
  A poset $\Cpoone = (\cpoone, \cpoleq_\cpoone)$ is an $\omega$-\emph{complete
  partial order},
  \ocpo\ for short, if any $\omega$-chain $(x_n)_{n<\omega}$ in $\cpoone$ has least 
  upper bound (lub) in $\cpoone$. A poset $\Cpoone = (\cpoone, \cpoleq)$ is an
  $\omega$-\emph{complete pointed partial order}, \ocppo\ for short, if it is
  an \ocpo\ with a least element $\bot_{\cpoone}$.
\end{definition}

  For an \ocppo\ $\Cpoone = (\cpoone, \cpoleq_{\cpoone}, \bot_{\cpoone})$ we
  will often omit subscripts, thus writing $\cpoleq$ and $\bot$ for
  $\cpoleq_{\cpoone}$ and $\bot_{\cpoone}$, respectively. Given an
  $\omega$-chain $(x_n)_{n < \omega}$ we will denote its least upper bound by
  $\lub_{\cpoone} \{x_n \mid n < \omega\}$. Oftentimes, following the above
  convention, we will shorten the latter notation to $\lub_{n < \omega} x_n$.

  Notice that an $\omega$-chain being a sequence, its elements need not be
  distinct. In particular, we say that the chain is \emph{stationary} if there
  exists $N < \omega$ such that $x_{N + n} = x_N$, for any $n < \omega$.

  We will often use the following basics result, stating that if we discard
  any finite number of elements at the beginning of a chain, we do not affect
  its set of upper bounds (and its lub).
\begin{lemma}\label{lub-from-bigger-index}
  For any $\omega$-chain $(x_n)_{n<\omega}$ and $N < \omega$ the following
  equality holds: $$\lub_{n < \omega} x_n = \lub_{n < \omega} x_{N+n}.$$
\end{lemma}

\begin{definition}
  Let $\Cpoone = (\cpoone, \cpoleq_{\cpoone}, \bot_{\cpoone}),
  \Cpotwo = (\cpotwo, \cpoleq_{\cpotwo}, \bot_{\cpotwo})$ be \ocppo s. 
  We say that a function $f : \cpoone \to \cpotwo$ is \emph{monotone} if for
  all elements $x,y$ in $\cpoone$, $x \cpoleq_{\cpoone} y$ implies $f(x)
  \cpoleq_{\cpotwo} f(y)$. We say it is \emph{continuous} if it is monotone
  and preserves lubs. That is, for any $\omega$-chain $(x_n)_{n<\omega}$ in
  $\cpoone$ we have: $$f(\lub_{\cpoone} \{x_n \mid n<\omega\}) =
  \lub_{\cpotwo} \{f(x_n) \mid n<\omega\}.$$ Finally, we say it is
  \emph{strict} if $f(\bot_\cpoone) = \bot_\cpotwo.$
\end{definition}

\begin{remark} (see \cite{AbramskyJung/DomainTheory/1994})
  It can be easily shown that if a function is continuous then it is also
  monotone. However, it should be noticed that to prove that a function $f :
  \cpoone \to \cpotwo$ is continuous, it is necessary to prove that
  $(f(x_n))_{n<\omega}$ forms an $\omega$-chain in $\cpotwo$, for any
  $\omega$-chain $(x_n)_{n<\omega}$ in $\cpoone$. That is equivalent to prove
  monotonicity of $f$.
\end{remark}

We denote the set of continuous functions from $\cpoone$ to $\cpotwo$ by
$\contFun{\cpoone}{\cpotwo}$ and write $f : \contFun{\cpoone}{\cpotwo}$ for $f
\in \contFun{\cpoone}{\cpotwo}$.

We will implicitly use the fact that continuous endofunctions on \ocppo s are
guaranteed to have least fixed points: given a continuous endofunction $f :
\cpoone \to \cpoone$ on an \ocppo\ $\cpoone$, there exists an element $\lfp{f}
\in \cpoone$ such that $f(\lfp{f}) = \lfp{f}$, and for any $x \in \cpoone$, if
$f(x) = x$ then $\lfp{f} \cpoleq x$. Throughout this work, we will use the
notation $\lfp{f}$ and $\gfp{f}$ to denote least and greatest fixed point of a
function $f$, respectively.

\ocppo s and continuous functions form a category, $\ocppoCat$, which has a cartesian closed 
structure. In particular, the cartesian product (of the underlying sets) of
\ocppo s is an
\ocppo\ when endowed with the pointwise order (with lubs and bottom element computed pointwise). 
Similarly, the set of continuous functions spaces between \ocppo s is an
\ocppo\ when endowed with the pointwise order (again, with lubs and bottom
element computed pointwise)\footnote{ Let $\Cpoone, \Cpotwo$ be \ocppo s
define:
\begin{itemize}
\item The \ocppo\ structure on $D \times E$is given by: 
    \begin{eqnarray*}
      (x,y) \cpoleq_{D \times E} (x',y')  
          &\Leftrightarrow&  x \cpoleq_D x' \wedge y \cpoleq_{E} y';  \\
      \bot_{\cpoone \times \cpotwo}
          &=& (\bot_{\cpoone}, \bot_{\cpotwo}) ; \\
      \lub_{\cpoone \times \cpotwo} \{(x_n,y_n) \mid n < \omega\}
          &=& (\lub_{\cpoone} \{x_n \mid n < \omega\}, \lub_{\cpotwo} \{y_n \mid n < \omega\}).
    \end{eqnarray*}
  \item The \ocppo\ structure on $\contFun{D}{E}$ is given by:
    \begin{eqnarray*} 
      f \cpoleq_{\contFun{D}{E}} g 
        &\Leftrightarrow& \forall x \in D.\ f(x) \cpoleq_E g(x) ; \\
      \bot_{\contFun{\cpoone}{\cpotwo}} 
        & = & \fun{x}{\bot_E} ; \\
      \lub_{\contFun{\cpoone}{\cpotwo}} \{f_n \mid n < \omega\} 
        & = &\fun{x}{\lub_{\cpotwo} \{f_n(x) \mid n < \omega\}}.
    \end{eqnarray*}
    \end{itemize}
}.
Notice that the function space $\contFun{\cpoone}{\cpotwo}$ between \ocppo s
$\cpoone$ and $\cpotwo$ is an \ocppo\ even if $D$ does not have a least
element (i.e. if $D$ is an \ocpo). As a consequence, since we can regard any
set $X$ as an \ocpo\ ordered by the identity relation $=_X$ on $X$\footnote{We
call such \ocpo s \emph{discrete}}, the set $\contFun{X}{\cpoone} = X \to
\cpoone$ is always an \ocppo, for any \ocppo\ $\cpoone$. The following well-
known result will be useful in several examples.
\begin{lemma}\label{continuous-in-both-arguments}
Let $\Cpothree, \Cpoone, \Cpotwo$ be \ocppo s. A function $f : C \times D \to
E$ is:
\begin{varenumerate}
  \item Monotone, if it is monotone in each arguments separately.
  \item Continuous, if it is continuous in each arguments separately.
\end{varenumerate}
\end{lemma} 

Following \cite{PlotkinPower/FOSSACS/01,PlotkinPower/FOSSACS/02}, we consider
operations (like $\oplus$ or $\mathsf{raise}_e$ in the examples from
Section~\ref{sect:informal}) form a given signature as sources of effects.
Semantically, dealing with operation symbols requires the introduction of
appropriate algebraic structures interpreting such operation symbols as
suitable functions. Combining the algebraic and the order theoretic structures
just described, leads to consider algebras carrying a domain structure
(\ocppo, in this paper), such that \emph{all} function symbols are interpreted
as continuous functions. The formal notion capturing all these desiderata is
the one of a continuous $\signature$-algebra \cite{Goguen/1977}.

Recall that a signature $\signature = (\mathcal{F}, \alpha)$ consists of a set
$\mathcal{F}$ of operation symbols and a map $\alpha : \mathcal{F} \to
\mathbb{N}$, assigning to each operation symbol a (finite) arity. A
$\signature$-algebra $(A, (\cdot)^{A})$ is given by a carrier set $A$ and an
interpretation $(\cdot)^{A}$ of the operation symbols, in the sense that for
$\op \in \mathcal{F}$, $\op^{A}$ is a map from $A^{\alpha(\op)}$ to $A$. We
will write $\op \in \signature$ for $\signature = (\mathcal{F}, \alpha)$ and
$\op \in \mathcal{F}$.

\begin{definition} 
  Given a signature $\signature$, a continuous $\signature$-algebra is an
  \ocppo\ $\Cpoone = (\cpoone, \cpoleq, \bot)$ such that for any function 
  symbol $\op$ in $\signature$ there is an associated continuous function $\op^D :
  D^{\alpha(\op)} \to D$.
\end{definition}

\begin{remark}
  Observe that for a function symbol $\op \in \signature$, 
  we do not require $\op^D$ to be strict.
\end{remark}

Before looking at monads, we now give various examples of
concrete algebras which can be given the structure of a continuous 
$\signature$-algebra for certain signatures. This testifies the applicability 
of our theory to a relatively wide range of effects. 
 
\begin{example}\label{omegaCppo} 
  Let $X$ be a set: the following are examples of \ocppo.
  \begin{varitemize}
  \item 
    The flat lifting $X_\bot$ of $X$, defined as $X + \{\bot\}$, ordered as
    follows: $x \cpoleq y$ iff $x = \bot$ or $x=y$.
  \item The set $(X + E)_\bot$ (think to $E$ as a set of exceptions), ordered as 
    in the previous example. We can consider the signature $\signature =
    \{\raiseExcep{e} \mid e \in E\}$, where each operation symbol
    $\raiseExcep{e}$ is interpreted as the constant $\inl(\inr(e))$.
  \item 
    The powerset $\powerset{X}$, ordered by inclusion. The least upper bound
    of a chain of sets is their union, whereas the bottom is the empty set. We
    can consider the signature $\signature = \{\oplus\}$ containing a binary
    operation symbol for nondeterministic choice. The latter can be
    interpreted as (binary) union, which is indeed continuous.
  \item 
    The set of subdistributions $\dist{X} = \{\distone : X \to [0,1]
    \mid \support{\distone} \text{ countable, } \sum_{x \in X}
    \distone(x) \leq 1\}$ over $X$, ordered
    pointwise: $\distone \cpoleq \disttwo$ iff $\forall x \in
    X.\ \distone(x) \leq \disttwo(x)$.  Note that requiring the
    support of $\distone$ to be countable is equivalent to requiring
    the existence of $\sum_{x \in X} \distone(x)$.  The \ocppo\
    structure is pointwise induced by the one of $[0,1]$ with the
    natural ordering. The least element is the always zero
    distribution $\fun{x}{0}$ (note that the latter is a
    subdistribution, and not a distribution). We can consider the
    signature $\signature = \{\oplus_p \mid p \in [0,1] \}$ with a
    family of probabilistic choice operations indexed by real numbers
    in $[0,1]$. We can interpret $\oplus_p$ as the binary operation 
    $\fun{(x,y)}{p \cdot x + (1 - p) \cdot y}$, which is indeed continuous.
  \item 
    The set $(S \times X)^S_\bot$, or equivalently $S \rightharpoonup (X
    \times S)$ (the set of partial states over $X$) with extension order: $f
    \cpoleq g$ iff $\forall x \in X.\ f(x) \neq
    \bot \Rightarrow f(x) = g(x)$, for a fixed set $S$ (of states). 
    The bottom element is the totally undefined function $\fun{x}{\bot}$,
    whereas the least upper bound of a chain $(f_n)_{n < \omega}$ is computed
    pointwise. Depending on the choice of $S$, we can define several
    continuous operations on $(S \times X)^S_\bot$. For instance, taking $S =
    \{\mathtt{true},\mathtt{false}\}$, the set of booleans, we can consider
    the signature $\signature =
    \{\mathsf{read}, \mathsf{write}_b \mid b \in S\}$ to be
    interpreted as the continuous operations $\mathit{read}$ and
    $\mathit{write}_b$ defined by
    \begin{align*}
      \mathit{write}_b(f)    & =  \fun{x}{f(b)};   \\
      \mathit{read}(f,g)     & =  \fun{x}{\text{if } x = \mathsf{true} \text{ then } f(x) \text{ else } g(x)}.
    \end{align*}  
  \item 
    The set $\stream{U} \times X_{\bot}$ (modelling computations with output
    streams) with the product order, for a fixed set $U$ (think of $U$ as a
    set of characters). The set $\stream{U}$ of streams over $U$ is the set of
    all finite and infinite strings (or words) over $U$. Formally, a stream $u
    \in \stream{U}$ is a function $u :
    \mathbb{N} \to U_\bot$ (i.e. a partial function from $\mathbb{N}$ to $U$) 
    such that $u(n) = \bot$ implies $u(n+c) = \bot$, for any $c \geq 0$. A
    finite stream is a function $u$ such that there exists an $n$ for which
    $u(n) =
    \bot$. We can endow $\stream{U}$ with the so-called approximation order, i.e. 
    the extension order on $\mathbb{N} \to U_{\bot}$. A finite approximation
    of length $n$ of a stream $u$ is a stream $w$ of length $n$ such that $w
    \cpoleq u$ holds. Clearly, the set of finite approximants of a stream $u$
    forms an $\omega$-chain, and for any $u
    \in \stream{U}$ we have $u = \lub_{n < \omega} u^{(n)}$,
    where $u^{(n)} = (u(0), \hh, u(n-1), \bot)$ denotes the $n$-th approximant
    of $u$. We can define the concatenation $\concat{u}{w}$ of a finite stream
    $u = (u(0), \hh, u(n-1), \bot)$ and a stream $w \in \stream{U}$ by
    $$
    (\concat{u}{w})(k) = 
        \begin{cases}
          u(k)  & \text{if } k \leq n - 1;    \\
          w(c)  & \text{if } k = n+c,\text{ for } c \geq 0.
        \end{cases}
    $$
    We can extend concatenation to infinite streams
    defining $\concat{u}{w} = u$, for $u$ infinite stream. 
    It is easy to prove that concatenation is continuous in
    its second argument, although even monotonicity fails for its
    first argument. For, consider the streams $c, cc$ (which are
    shorthand for $(c, \bot), (c,c,\bot)$, respectively). We
    clearly have $c \cpoleq cc$, but $\concat{c}{b} = cb \not \cpoleq
    ccb = \concat{cc}{b}$. \\
    Finally, we can consider the signature
    $\signature = \{\mathsf{print}_c \mid c \in U\}$ interpreted as the
    family of operations $\mathit{print}_c$ defined by $
    \mathit{print}_c(u, x) = (\concat{c}{u}, x)$. It is easy to 
    see that since concatenation is continuous in its second argument, 
    then so does $\mathit{print}_c$.
\end{varitemize}
\end{example}
\subsection{Monads}

The notion of monad is given via the equivalent notion of Kleisli Triple (see
\cite{MacLane/Book/1971}). Let $\mathbb{C}$ be a category.

\begin{definition}
  A Kleisli Triple $\lan \monad, \uniT, \kleisli{(\cdot)}\ran$ consists of an
  endomap $\monad$ over objects of $\mathbb{C}$, a family of arrows $\uniT_X$,
  for any object $X$, and an operation (called Kleisli extension or Kleisli
  star) $\kleisli{(\cdot)} : \HomSet{\mathbb{C}}{X}{\monad Y} \to
  \HomSet{\mathbb{C}}{\monad X}{\monad Y}$,
  (for all objects $X,Y$) satisfying the equations
 \begin{align*}
    \kleisli{f} \circ \uniT         & =  f;   \\
    \kleisli{\uniT}                 & =  id ; \\
    \kleisli{(\kleisli{g} \circ f)} & =  \kleisli{g} \circ \kleisli{f};
  \end{align*}
  where $f$ and $g$ have the appropriate types.
\end{definition}
Given the equivalence between the notions of monad and Kleisli Triple, we will
be terminologically sloppy, using the terms `monads' and `Kleisli Triples'
interchangeably. In particular, for a monad/Kleisli Triple $\lan \monad,
\uniT, \kleisli{(\cdot)}\ran$ we will implicitly assume functoriality of the
endomap $\monad$. Finally, we will often denote a Kleisli Triple $\lan \monad,
\uniT, \kleisli{(\cdot)}\ran$ simply as $\monad$.

To any Kleisli Triple $\lan \monad, \uniT, \kleisli{(\cdot)}\ran$ on a
category $\mathbb{C}$ we can associate the so-called Kleisli category
$\kleisliCat{\monad}$ over $\mathbb{C}$.

\begin{definition}
  Given a Kleisli triple as above, we define the Kleisli category $\kleisliCat{\monad}$ 
  (over $\mathbb{C}$) as follows: 
  \begin{varitemize}
    \item Objects of $\kleisliCat{\monad}$ are those of $\mathbb{C}$.
    \item To any arrow $f : X \to \monad Y$ in $\mathbb{C}$ we associate an 
      arrow $\bar f : X \to Y$ in $\kleisliCat{\monad}$.
    \item The identity arrow $id_X : X \to X$ in $\kleisliCat{\monad}$ is $\uniT_X$.
    \item Given arrows $\bar f : X \to Y, \bar g : Y \to Z$ (which correspond to arrows 
      $f : X \to \monad Y$ and $g : Y \to \monad Z$ in $\mathbb{C}$), define
      their composition to be $\kleisli{g} \circ f$.
  \end{varitemize}
\end{definition}

From now on we fix the base category $\mathbb{C}$ to be the category $\set$ of sets and functions.

\begin{remark}
  Since we work in $\set$, we will extensively use the so called bind operator
  $\bind$ in place of Kleisli extensions. Such operator takes as arguments an
  element $u$ of $\monad X$, together with a function $f : X \to \monad Y$ and
  returns an element $u \bind f$ in $\monad Y$. Concretely, we can define $u
  \bind f$ as $\kleisli{f}(u).$ Vice versa, we can define the Kleisli
  extension $\kleisli{f}$ of $f$ as $\fun{x}{(x \bind f)}$.
\end{remark}

\begin{example}\label{monadExample} 
  All the constructions introduced in Example \ref{omegaCppo} carry the 
  structure of a monad. 
  \begin{varitemize}
  \item The functor $\monad X = X_{\bot}$ is (part of) a monad, with left injection 
    as unit and bind operator defined by
    $$
      u \bind f = 
      \begin{cases}
        f(x)        & \text{if } u = \inl(x), \text{ for some }x\in X; \\
        \inr(\bot)  & \text{otherwise}.
      \end{cases}
    $$
  \item 
    The powerset functor $\powersetmonad$ is a monad with unit
    $\fun{x}{\{x\}}$ and bind operator defined by $u \bind f =
    \bigcup_{x \in u} f(x)$.
  \item 
    The subdistribution functor $\distribution$ is a monad 
    with unit given via the Dirac distribution $\delta$ and bind operator defined by
    $$\distone \bind f = \fun{y}{\sum_{x \in X} \distone(x) \cdot f(x)(y)}.$$
  \item 
    The partiality and exception functor $\monad X = (X +
    E)_{\bot}$ for a given set $E$ of exceptions is a monad with the
    function $\fun{x}{\inl(\inl(x))}$ as unit. The bind operator is defined by
    $$
    u \bind f = 
    \begin{cases}
       u      & \text{if } u = \inr(\bot) \text{ or } u = \inl(\inr(e)); \\
     f(x)     & \text{if } u = \inl(\inl(x)).
    \end{cases}
    $$
  \item 
    The partiality and global state functor $\monad X = S
    \to (X \times S)_{\bot}$ for a given set $S$ of states, is
    a monad with unit $\fun{x}{(\fun{s}{(x,s)})}$ and the bind
    operator defined by
    $$(\sigma \bind f)(s) = 
    \begin{cases}
      \inr(\bot)    & \text{ if } \sigma(s) = \inr(\bot);    \\
      f(y)(t)       & \text{ if } \sigma(s) = \inr(y,t).
    \end{cases}
    $$
    \item The output functor $\monad X = \stream{U} \times X_\bot$ is a monad 
      with unit $\fun{x}{(\varepsilon, \inl(x))}$ where, to avoid confusion, we 
      use denote the empty stream by $\varepsilon$, and bind operator defined by
      \begin{align*}
        (u, \inr(\bot)) \bind f   &= (u, \inr(\bot));  \\
        (u,\inl(x)) \bind f       &= (\concat{u}{w}, y)  \\
                                  &\quad  \text{ where } (w,y) = f(x).
      \end{align*}
  \end{varitemize}
\end{example}

For a given signature $\signature$, we are interested in monads on $\set$ 
that carry a continuous $\signature$-algebra structure.
\begin{definition}\label{ocppo-order}
  An \emph{\ocppo\ order} $\cpoleq$ on a monad $\monad$ is a map that
  assigns to each set $X$ a relation $\cpoleq_X \subseteq \monad X
  \times \monad X$ and an element $\bot_X \in \monad X$ such that
  \begin{varitemize}
  \item 
    The structure $(\monad X, \cpoleq_X, \bot_X)$ is an \ocppo.
  \item The bind operator is continuous in both arguments. That is,
    \begin{align*}
      (\lub_{n<\omega} u_n) \bind f  & =  \lub_{n<\omega} (u_n \bind f);   \\
      u \bind (\lub_{n<\omega} f_n)  & =  \lub_{n<\omega} (u \bind f_n).
    \end{align*}
  \end{varitemize}
  We say that $\cpoleq$ is strict in its first argument if we additionally
  have $\bot \bind f = \bot$ (and similarly for its second argument). We say
  that $\monad$ carries a continuous $\signature$-algebra structure if
  $\monad$ has an \ocppo\ order such that $\monad X$ is a continuous
  $\signature$-algebra with respect to the order $\cpoleq_X$, for any set $X$.
\end{definition}
Most of the time we will work with a fixed set $X$. As a consequence, we will
omit subscripts, just writing $\cpoleq$ in place of $\cpoleq_X$. Similarly,
for an operation $\op$ in $\signature$, we will write $\op^\monad$ in place of
$\op^{\monad X}$ (the interpretation of $\op$ as an operation on $\monad X$).

\begin{remark}
  The last definition is essentially regarding the bind operator as a
  \emph{continuous} function (in both arguments) from $\monad X\times (X
  \to\monad Y)$ to $\monad Y$. This makes sense since $\monad X \times (X \to\monad Y)$ is an 
  \ocppo: regarding the set $X$ as the discrete \ocpo, we have 
  $X \to \monad Y = \contFun{X}{\monad Y}$, so that $\monad X \times (X \to
  \monad Y)$ is an \ocppo, being the product of two \ocppo s. Because $\bind$
  is continuous in both its arguments, we have $(\lub_{n} u_n) \bind (\lub_{n}
  f_n) = \lub_{n} (u_n \bind f_n)$.
\end{remark}

The bind operation will be useful when giving an operational semantics to the
sequential (monadic) composition of programs. As a consequence, although we
did not explicitly require the bind operator to be strict (especially in its
first argument), such condition will be often desired (especially when giving
semantics to call-by-value languages).

\begin{example}\label{outputDoesNotWork}
  Example \ref{omegaCppo} shows that all monads in Example \ref{monadExample}
  have an \ocppo\ order. It is easy to check that all bind operations, with
  the exception of the one for the output monad, are strict in their first
  argument. In fact, even monotonicity of the bind operator for output monad
  fails, due to the failure of monotonicity for concatenation (see
  \ref{omegaCppo}). The reason why this property does not hold for the output
  monad relies on non-monotonicity of the concatenation operator on its first
  argument. Nonetheless, we can endow $U^{\infty} \times X_{\bot}$ with a
  different order, obtaining the desired result:
  $$(u,x) \cpoleq (w,y)\ \text{ iff }\ (x = \inr(\bot) \wedge u \cpoleq w)
  \vee (x \neq \inr(\bot) \wedge x = y \wedge u = w).
  $$
  It is not hard to see that we obtain an \ocppo\ with continuous bind
  operator.
\end{example}

Definition \ref{ocppo-order} requires the bind operator to be continuous. 
This condition is a special case of the more general notion of order-enrichment 
\cite{Kelly/EnrichedCats} for a monad.
  \begin{definition}
    We say that a category $\mathbb{C}$ is \ocppo-enriched if 
    \begin{varitemize}
      \item Each hom-set $\HomSet{\mathbb{C}}{X}{Y}$ carries a partial order
        $\cpoleq$ with an \ocppo\ structure.
      \item Composition is continuous. That is, the following equations hold: 
        \begin{align*}
          g \circ (\lub_{n < \omega} f_n)   & =  \lub_{n < \omega} (g \circ f_n);   \\
          (\lub_{n < \omega} f_n) \circ g   & =  \lub_{n < \omega} (f_n \circ g).
        \end{align*}
    \end{varitemize}
  \end{definition}

  \begin{definition}
  A monad $\monad$ on $\mathbb{C}$ is \ocppo-enriched if $\kleisliCat{\monad}$
  is
  \ocppo-enriched. That is, for every pair of objects $X, Y$, the set 
  $\HomSet{\mathbb{C}}{X}{\monad Y}$ carries an \ocppo-structure such that
  composition is continuous and Kleisli star is locally continuous.
  Concretely, that means that the following equations hold (cf.
  \cite{GoncharovSchroder/LICS/2013}):
  \begin{align*}
    (\lub_{n<\omega} f_n) \circ h          & = \lub_{n<\omega} (f_n \circ h);             \\ 
    \kleisli{u} \circ \lub_{n<\omega} f_n  & = \lub_{n<\omega} (\kleisli{u} \circ f_n);   \\
    \kleisli{(\lub_{n<\omega} f_n)}        & = \lub_{n<\omega} \kleisli{f_n}.
  \end{align*}
\end{definition} 
Our notion of \ocppo\ order on a monad $\monad$ on $\set$ is nothing but a
special case of \ocppo-enrichment. Since we are in $\set$, and we have the
terminal object $1$ (say $1 = \{*\}$), any element $u$ of $\monad X$ directly
corresponds to the arrow $\bar u: 1 \to \monad X$, defined by $\bar u(*) = u$.
In particular, we have $\monad X \cong 1 \to \monad X = \contFun{1}{\monad X}$
(since $1$ is discrete). For a function $f : X \to Y$ and an element $u \in X$
we can simulate function application $f(u)$ as $\bar u \circ f$ (meaning that
$\bar u \circ f = \overline{f(u)}$). As a consequence, we have that $u \bind
f$ corresponds to $\kleisli{f} \circ \bar u$. Finally, observe that the
equation $$\overline{\lub_{\monad X} \{u_n \mid n < \omega\}} = \lub_{1 \to
\monad X} \{\overline{u}_n \mid n < \omega\}$$ holds. We show that if $\monad$
is \ocppo-enriched, then the bind operator is continuous in both arguments. In
fact, $(\lub_{n < \omega} u_n) \bind f$ corresponds to the function
$$\kleisli{f} \circ \overline{\lub_{n < \omega} u_n} =
\kleisli{f} \circ \lub_{n < \omega} \bar{u}_n = 
\lub_{n < \omega} (\kleisli{f} \circ \bar{u}_n)$$
which itself corresponds to $\lub_{n< \omega} (u_n \bind f)$. Similarly, $u
\bind \lub_{n<\omega} f_n$ corresponds to $$\kleisli{(\lub_{n<\omega} f_n)}
\circ \bar u =
\lub_{n<\omega} \kleisli{f_n} \circ \bar u = 
\lub_{n<\omega} (\kleisli{f_n} \circ \bar u)$$
which corresponds to $\lub_{n<\omega} (u \bind f_n)$.

Most of the monads commonly used e.g. in functional programming 
to model side-effects are not order enriched. This follows from 
the requirement of having a bottom element. 
The reason behind that condition relies on the fact
that our operational semantics will be model non-termination
explicitly. That is, a (purely) divergent program $\termone$ will be
evaluated in the bottom element of the monad. For instance let us
consider pure $\lambda$-calculus. Standard operational semantics
employs inductively defined judgments of the form
$\evalToo{\termone}{\valone}$, meaning that $\evalToo{}{} \subseteq
\terms \times \values$ (which, in general, can be viewed as
$\evalToo{}{} \subseteq \terms \times \monad \values$, for $\monad$ 
the identity monad). Such a semantics does not capture divergence explicitly:
for instance, we just have that there exists no value $\valone$ such
that $\evalToo{\Omega}{\valone}$. The operational semantics we will
define in the next chapter associates to each program a subset of the
finite approximations it is evaluated to, and then consider the lub of
such approximations. As a consequence, we need the monad to have
bottom element $\bot$, so that we will have that the semantics of
$\Omega$ to be indeed $\bot$.  Nevertheless, we recall that any set 
can always be lifted to an \ocppo\ by adding a bottom element to it and 
considering the flat ordering.  As a consequence, although most of the 
monads commonly used in functional programming are not order-enriched, 
their flat version is.

\section{A Computational Calculus and Its Operational Semantics}
In this section we define a computational $\lambda$-calculus.
Following \cite{Moggi/LICS/89,Lassen/PhDThesis,Levy/InfComp/2003}, we syntactically 
distinguish between values and computations. 
We fix a signature $\signature$ of operation symbols (the sources of side-effects), 
and a monad $\monad$ carrying a continuous 
$\signature$-algebra structure (which describes the nature of the wanted effectful 
computations generated by the operations in $\signature$). 

\begin{definition}
  Given a signature $\signature$, the sets $\terms_\signature$ and $\values_\signature$ of 
  terms and values are defined by the following grammars:
\begin{eqnarray*}
  \termone, \termtwo  & ::= & \return{\valone}          \mid
                              \valone \valtwo           \mid 
                              \seq{\termone}{\termtwo}  \mid
                              \op(\termone, \hh, \termone); \\
  \valone, \valtwo    & ::= & \varone                   \mid
                              \abs{\varone}{\termone}.
\end{eqnarray*}
where $\varone$ ranges over a fixed countably infinite set
$\variables$ of variables and $\op$ ranges over $\signature$.
\end{definition}
The term $(\seq{\termone}{\termtwo})$ captures monadic binding (which
is usually expressed using a ``let-in'' notation).  A calculus with an
explicit separation between terms and values has the advantage to make
proofs simpler, without sacrificing expressiveness. For instance, 
we can encode terms' application $\termone \termtwo$ as 
$(\seq{\termone}{(\termtwo\ \mathsf{to}\ \vartwo. \varone\vartwo)})$ and
vice versa $(\seq{\termone}{\termtwo})$ as $(\abs{\varone}{\termtwo})\termone$.

\begin{example}
  We can model several calculi combining the signatures from Example \ref{omegaCppo}.
  \begin{varitemize}
  \item 
    For a given set $E$ of exceptions, we can define a probabilistic 
    $\lambda$-calculus with exceptions
    as $\terms_\signature$, for a signature 
    $\signature = \{\oplus_p, \raiseExcep{e} \mid p \in [0,1], e \in E\}$. 
    In particular, we will have terms of the form $
    \termone \oplus_p \termtwo$  and $\raiseExcep{e}$. 
    Replacing the probabilistic choice operator $\oplus_p$ with its
    nondeterministic counterpart $\oplus$ we obtain a nondeterministic
    calculus with exceptions.
  \item 
    We can define a nondeterministic calculus with global (boolean) states as
    $\terms_\signature$, for a signature $\signature = \{\oplus,
    \mathsf{write}_b, \mathsf{read} \mid b \in \{\mathit{true},
    \mathit{false}\} \}$. In particular, we will have terms of the form 
    $\termone \oplus \termtwo$,  $\mathsf{write}_b.\termone$, and
    $\mathsf{read}(\termone, \termtwo)$. The intuitive meaning of
    $\mathsf{write}_b.\termone$ is to store $b$ and then continue as
    $\termone$, whereas the intuitive meaning of $\mathsf{read}(\termone,
    \termtwo)$ is to read the value in the store: if such value is the boolen
    true then continue as $\termone$, otherwise as $\termtwo$. A formal
    semantics for these two functions is given in Example \ref{omegaCppo}.
  \item 
    We can define a nondeterministic calculus with output using the signature
    $\signature = \{\oplus, \mathsf{print}_c \mid c \in U\}$, where $U$ is a
    given alphabet. The intuitive meaning of $\print{c}{\termone}$ is to
    output $c$ and then continue as $\termone$. A formal semantics for this
    function is given in Example \ref{omegaCppo}.
  \end{varitemize}
\end{example}

In what follows, we work with a fixed arbitrary signature $\signature$. As a
consequence, we often denote the sets of terms and values as $\terms$ and
$\values$, respectively, thus omitting subscripts. Moreover, we consider terms
and values modulo $\alpha$-equivalence and assume Barendregt Convention
\cite{Barendregt/Book/1984}. We let $\fv{\termone}$ denote the set of free
variables of the term $\termone$. A term $\termone$ is closed if
$\fv{\termone} =
\emptyset$. We denote finite sets of variables, terms and values using ``bar notation'': 
for instance, we write $\vars$ and $\bar{\valone}$ for a finite set of
variables and values, respectively. For a finite set $\vars$ of variables
define
\begin{align*}
   \terms(\vars)    & = \{\termone \mid \fv{\termone} \subseteq \vars\};  \\
   \values(\vars)   & = \{\valone \mid \fv{\valone} \subseteq \vars\};
\end{align*}
to be the sets of terms and values with free variables in $\vars$, respectively. 
The set of closed terms and values are then defined as $\terms(\emptyset)$
and $\values(\emptyset)$, and denoted by $\closedTerms$ and $\closedValues$, 
respectively.

\begin{definition}
  Define for all values $\valone, \valtwo$ and any term $\termone$, the value 
  $\subst{\valone}{\valtwo}{\vartwo}$ obtained by 
  (simultaneous) substitution of $\valtwo$ for $\vartwo$ in $\valone$, and the 
  term $\substcomp{\termone}{\vartwo}{\valone}$ obtained by (simultaneous) 
  substitution of $\valone$ for $\vartwo$ in $\termone$ as follows 
  (recall we are assuming Barendregt's convention):
  \begin{eqnarray*}
    \subst{\varone}{\valtwo}{\varone}             & = & 
    \valtwo   \\  
    \subst{\varone}{\valtwo}{\vartwo}           & = & 
    \varone   \\
    \subst{(\abs{\varone}{\termone})}{\valtwo}{\vartwo}   & = & 
    \abs{\varone}{\substcomp{\termone}{\vartwo}{\valtwo}} \\
    & & \\
    \substcomp{(\return{\valone})}{\vartwo}{\valtwo}      & = & 
    \return{\subst{\valone}{\valtwo}{\vartwo}}  \\
    \substcomp{(\valone \valone')}{\vartwo}{\valtwo}      & = & 
    \subst{\valone}{\valtwo}{\vartwo} \subst{\valone'}{\valtwo}{\vartwo} \\
    \substcomp{(\seq{\termone}{\termtwo})}{\vartwo}{\valtwo}  & = & 
    \seq{\substcomp{\termone}{\vartwo}{\valtwo}}{\substcomp{\termtwo}{\vartwo}{\valtwo}}
  \end{eqnarray*}
\end{definition}

Big-step semantics associates to each closed term $\termone$ an element
$\sem{\termone}$ in $\monad \closedValues$. Such a semantics is defined by
means of an approximation relation $\evalTo{}{n}{}$, indexed by a natural
number $n$, whose definition is given in Figure \ref{bigStepSemantics}.
Judgments are of the form $\evalTo{\termone}{n}{\compone}$, where $\termone
\in \closedTerms$, $\compone \in \monad \closedValues$ and $n \geq 0$.
Intuitively, a judgment $\evalTo{\termone}{n}{\compone}$ states that
$\compone$ is the $n$-th approximation of the computation obtained by call-by-value 
evaluating $\termone$. (By the way, all the results in this paper would
remain valid also if evaluating terms in call-by-\emph{name} order, which is
however less natural
in presence of effects.)

 \begin{figure}[htbp]
 \begin{mdframed}
  \begin{center}
    \begin{tabular}{c}
      \AXC{\phantom{n>0}} \RL{$\ruleBot$}
      \UIC{$\evalTo{\termone}{0}{\bot}$} \DP 
      $\quad$
      \AXC{\phantom{$n > 0$}} \RL{$\ruleReturn$}
      \UIC{$\evalTo{\return \valone}{n+1}{\unit{\valone}}$}\DP 
      \\  \\
      \AXC{$\evalTo{\termone}{n}{\compone}$}
      \AXC{$\evalTo{\substcomp{\termtwo}{\varone}{\valone}}{n}{\comptwo_{\valone}}$} \RL{$\ruleSeq$}
      \BIC{$\evalTo{\seq{\termone}{\termtwo}}{n+1}{\compone \bind (\fun{\valone}{\comptwo_{\valone}})}$} \DP
      \\  \\
      \AXC{$\evalTo{\substcomp{\termone}{\varone}{\valone}}{n}{\compone}$} \RL{$\ruleApp$}
      \UIC{$\evalTo{(\abs{\varone}{\termone})\valone}{n+1}{\compone}$} \DP
      \\  \\
      \AXC{$\evalTo{\termone_1}{n}{\compone_1}$}
      \AXC{$ \hh $}
      \AXC{$\evalTo{\termone_k}{n}{\compone_k}$} \RL{$\ruleOp$}
      \TIC{$\evalTo{\op(\termone_1, \hh, \termone_k)}{n+1}{\op^{\monad}(\compone_1, \hh, \compone_k)}$}\DP          
  \end{tabular}
  \end{center}
  \end{mdframed}
  \caption{Big-step Semantics.}
  \label{bigStepSemantics}
\end{figure}

  The system in Figure \ref{bigStepSemantics} is `syntax directed', meaning that given a 
  judgment $\evalTo{\termone}{n}{X}$, the solely syntactic form of $\termone$ and the number $n$
  uniquely determine the last rule used to derive $\evalTo{\termone}{n}{X}$. As a consequence, 
  each judgment has a unique derivation. 

  \begin{lemma}[Determinacy]\label{determinacy}
    For any term $\termone$, if $\evalTo{\termone}{n}{X}$ and $\evalTo{\termone}{n}{Y}$, then 
    $X = Y$.
  \end{lemma}
  \begin{proof}
    By induction on $n$. If $n = 0$, then both $\evalTo{\termone}{n}{X}$ and
    $\evalTo{\termone}{n}{Y}$ must be the conclusion of an instance of rule
    $\ruleBot$ (all other rules requires $n$ to be positive). As a
    consequence, we have $X = \bot = Y$. Suppose now $n = m +1$, for some $m
    \geq 0$. We proceed by case analysis on the last rule used to derive
    $\evalTo{\termone}{n}{X}$.
    \begin{varitemize}
      \item[\textbf{Case} $\ruleBot$.] This case is not possible, since $n > 0$.
      \item[\textbf{Case} $\ruleReturn$.] Then $\termone$ is of the form
        $\return{\valone}$, for some value $\valone$, $X$ is $\unit{\valone}$,
        and $\evalTo{\termone}{m+1}{Y}$ is
        $\evalTo{\return{\valone}}{m+1}{Y}$. The latter judgment must follow
        from an instance of rule $\ruleReturn$ as well, and thus $Y =
        \unit{\valone}$.
      \item[\textbf{Case} $\ruleApp$.] Then $\termone$ is of the form
        $(\abs{\varone}{\termtwo})\valone$ and we have
        $\evalTo{\substcomp{\termtwo}{\varone}{\valone}}{m}{X}$, for some term
        $\termtwo$. Therefore, the judgment $\evalTo{\termone}{m+1}{Y}$ is of
        the form $\evalTo{(\abs{\varone}{\termtwo})\valone}{m+1}{Y}$ implying
        it can only be the conclusion of an instance of the rule $\ruleApp$.
        Therefore, $\evalTo{\substcomp{\termtwo}{\varone}{\valone}}{m}{Y}$
        holds as well. We can apply the induction hypothesis on the latter and
        $\evalTo{\substcomp{\termtwo}{\varone}{\valone}}{m}{X}$ thus inferring
        $X = Y$.
      \item[\textbf{Case} $\ruleSeq$.] Then $\termone$ is of the form
        $\seq{\termtwo}{\termtwo'}$, $X$ is of the form $X' \bind
        (\fun{\valone}{X_\valone'})$, and both $\evalTo{\termtwo}{m}{X'}$ and
        $\evalTo{\substcomp{\termtwo'}{\varone}{\valone}}{m}{X_{\valone}'}$
        hold, for some terms $\termtwo, \termtwo'$ and elements $X',
        X_{\valone}'$ in $\monad \closedValues$. As a consequence, the
        judgment $\evalTo{\termone}{m+1}{Y}$ has the form
        $\evalTo{\seq{\termtwo}{\termtwo'}}{m+1}{Y}$, implying it must be the
        conclusion of an instance of the rule $\ruleSeq$ as well. Therefore,
        we have $\evalTo{\termtwo}{m}{Y'}$ and
        $\evalTo{\substcomp{\termtwo'}{\varone}{\valone}}{m}{Y_{\valone}'}$,
        and $Y = Y' \bind (\fun{\valone}{Y_\valone'})$, for some elements $Y',
        Y_{\valone}'$. We can then apply the induction hypothesis on
        $\evalTo{\termtwo}{m}{X'}$, $\evalTo{\termtwo}{m}{Y'}$ and
        $\evalTo{\substcomp{\termtwo'}{\varone}{\valone}}{m}{X_{\valone}'}$,
        $\evalTo{\substcomp{\termtwo'}{\varone}{\valone}}{m}{Y_{\valone}'}$,
        obtaining $X' = Y'$, $X_\valone' = Y_\valone'$ and thus $X' \bind
        (\fun{\valone}{X_\valone'}) = Y' \bind (\fun{\valone}{Y_\valone'})$.
      \item[\textbf{Case} $\ruleOp$.] Then $\termone$ is of the form
        $\op(\termone_1, \hh, \termone_k)$, $X$ is of
        the form $\opT(X_1, \hh, X_k)$, and the judgment
        $\evalTo{\termone_1}{m}{X_1}, \hh, \evalTo{\termone_k}{m}{X_k}$ hold,
        for some terms $\termone_1, \hh, \termone_k$ and elements $X_1,\hh,
        X_k$ in $\monad \closedValues$. As a consequence, the judgment
        $\evalTo{\termone}{m+1}{Y}$ has the form $\evalTo{\opT(\termone_1,
        \hh, \termone_k)}{m+1}{Y}$. The latter must be the conclusion of an
        instance of the rule $\ruleOp$, meaning that we have judgments
        $\evalTo{\termone_1}{m}{Y_1},
        \hh, \evalTo{\termone_k}{m}{Y_m}$ and $Y = \opT(Y_1, \hh, Y_m)$. 
        We can apply the induction hypothesis on the pair of judgments 
        $\evalTo{\termone_i}{m}{X_i}, \evalTo{\termone_i}{m}{Y_i}$, 
        for $i \in \{1, \hh, k\}$, inferring
        $X_i = Y_i$. We conclude $\opT(X_1, \hh, X_k) = \opT(Y_1, \hh, Y_k)$.
    \end{varitemize} 
  \end{proof}

\begin{lemma}
  For any term $\termone$ if $\evalTo{\termone}{n}{X}$ and $\evalTo{\termone}{n+N}{Y}$, then 
  $X \cpoleq Y$.
\end{lemma}

\begin{proof}
  The proof follows the same pattern of the previous one, where in the inductive case we use monotonicity of 
  both the bind operator and the operations $\opT$.
\end{proof}

\begin{corollary}\label{approximants-of-M-form-omega-chain}
  Let $\termone$ be a a term and $X_n$ be the (unique) element in $\monad \closedValues$ such that 
  $\evalTo{\termone}{n}{X}$. Then, the sequence $(X_n)_{n<\omega}$ forms an $\omega$-chain in 
  $\monad \closedValues$.
\end{corollary}

A direct consequence of the above corollary is that we can define the
evaluation $\sem{\termone}$ of a term $\termone$ as
$$\sem{\termone} = \lub_{\evalTo{\termone}{n}{\compone}} \compone.$$
This allows us to explicitly
capture non-termination (which is usually defined coinductively). 
For instance, it is easy to show that for the
purely (i.e. having no side-effects) divergent program $\Omega$, 
defined as $(\abs{\varone}{\varone}{\varone})(\abs{\varone}{\varone}{\varone})$,
we have $\sem{\Omega} = \bot$. This style of operational semantics
\cite{DalLagoZorzi/TIA/2012,DalLagoSangiorgiAlberti/POPL/2014} is precisely the reason we require the
monad $\monad$ to carry an \ocppo\ structure. 
Modelling divergence in this way turned out
to be fundamental in e.g. probabilistic calculi 
\cite{DalLagoZorzi/TIA/2012}.

\begin{definition}
  Let $\termone$ be a term. Define the $n$-th approximation $\termone^{(n)} \in \monad \closedValues$ 
  of $\termone$ as follows:
  \begin{align*}
    \termone^{(0)}  
      & = \bot    \\
    (\return{\valone})^{(n+1)}
      & = \unit{\valone} \\
    ((\abs{\varone}{\termone})\valone)^{(n+1)}  
      &= (\substcomp{\termone}{\varone}{\valone})^{(n)} \\
    (\seq{\termone}{\termtwo})^{(n+1)}         
      &= \termone^{(n)} \bind (\fun{\valone}{(\substcomp{\termtwo}{\varone}{\valone})^{(n)}}) \\
    (\op(\termone_1, \hh, \termone_k))^{(n+1)} 
      &= \opT(\termone_1^{(n)}, \hh, \termone_k^{(n)})
  \end{align*}
\end{definition}

\begin{lemma}\label{approximantions-operational-semantics}
  For any term $\termone$ we have $\evalTo{\termone}{n}{\termone^{(n)}}$.
\end{lemma}

\begin{proof}
  The proof is by induction on $n$. If $n = 0$, then we trivially have $\evalTo{\termone}{0}{\bot}$. 
  If $n = m+1$, for some $m \geq 0$, we proceed by case analysis on the last rule used to derive 
  the judgment $\evalTo{\termone}{m+1}{\termone^{(m+1)}}$. As a paradigmatic example, we show the case 
  for rule $\ruleSeq$. Suppose $\evalTo{\termone}{m+1}{\termone^{(m+1)}}$ is of the form 
  $\evalTo{(\seq{\termtwo}{\termtwo'})}{m+1}{Y \bind (\fun{\valone}{Y_{\valone}'})}$ and the judgments 
  $\evalTo{\termtwo}{m}{Y}$, $\evalTo{\substcomp{\termtwo'}{\varone}{\valone}}{m}{Y_\valone'}$ hold, 
  for some terms $\termtwo, \termtwo'$ and elements $Y, Y_\valone'$ in $\monad \closedValues$. We can 
  apply the induction hypothesis on $m$, obtaining $\evalTo{\termtwo}{m}{\termtwo^{(m)}}$ and 
  $\evalTo{\substcomp{\termtwo'}{\varone}{\valone}}{m}{(\substcomp{\termtwo'}{\varone}{\valone})^{(m)}}$. 
  By Lemma \ref{determinacy} we thus have $\termtwo^{(m)} = Y$ and 
  $(\substcomp{\termtwo'}{\varone}{\valone})^{(m)} = Y_\valone'$. We can conclude 
  $$Y \bind (\fun{\valone}{Y_{\valone}'}) = 
  \termtwo^{(m)} \bind (\fun{\valone}{(\substcomp{\termtwo'}{\varone}{\valone})^{(m)}}) = 
  (\seq{\termtwo}{\termtwo'})^{(m+1)}.$$
\end{proof}
Corollary \ref{approximants-of-M-form-omega-chain} and Lemma \ref{approximantions-operational-semantics} 
together imply that for any term $\termone$ we have the $\omega$-chain $(\termone^{(n)})_{n < \omega}$ 
of finite approximations of $\termone$. That means, in particular, that $\sem{\termone}$ is equal to 
$\lub_{n < \omega} \termone^{(n)}$. 
For instance, by previous lemma we have $\Omega^{(0)} = \bot$ and $\Omega^{(n+1)} = 
(\substcomp{(\varone \varone)}{\varone}{\abs{\varone}{\varone}{\varone}})^{(n)} = 
\Omega^{(n)}$. As a consequence, we have for any $n$, $\Omega^{(n)} = \bot$, and thus $\sem{\Omega} = \bot$.

Since both $\bind$ and $\opT$ are continuous, we can characterise operational 
semantics equationally.
\begin{lemma}\label{semEquations}
  The following equations hold:
  \begin{align*}
    \sem{\return{\valone}}        
      &= \unit{\valone};  \\
    \sem{(\abs{\varone}{\termone})\valone}    
      &= \sem{\substcomp{\termone}{\varone}{\valone}};   \\
    \sem{\seq{\termone}{\termtwo}}
      &= \sem{\termone} \bind (\fun{\valone}{\sem{\substcomp{\termtwo}{\varone}{\valone}}}); \\
    \sem{\op(\termone_1, \hh, \termone_n)} 
      &= \opT(\sem{\termone_1}, \hh, \sem{\termone_n}).
  \end{align*}
\end{lemma}

\begin{proof}
  By Lemma \ref{lub-from-bigger-index} we have $\sem{\termone} = \lub_{n < \omega} \termone^{(n+1)}$, 
  meaning that we can freely ignore $\termone^{(0)}$ (which is $\bot$). We prove each equation 
  separately.
  \begin{varitemize}
    \item[\textbf{Case} 1.] We have:
      $$
      \sem{\return{\valone}}  = \lub_{n<\omega}(\return{\valone})^{(n+1)}  
                              = \lub_{n<\omega} \unit{\valone}             
                              = \unit{\valone}.
      $$
    \item[\textbf{Case} 2.] We have:
      $$
        \sem{(\abs{\varone}{\termone})\valone}  
            = \lub_{n<\omega} ((\abs{\varone}{\termone})\valone)^{(n+1)} 
            = \lub_{n<\omega} (\substcomp{\termone}{\varone}{\valone})^{(n)} 
            = \sem{\substcomp{\termone}{\varone}{\valone}}.
      $$
    \item[\textbf{Case} 3.] We have:
      \begin{align*}
        \sem{\seq{\termone}{\termtwo}}
          &= \lub_{n<\omega} (\seq{\termone}{\termtwo})^{(n+1)}   \\
          &= \lub_{n<\omega} (\termone^{(n)} \bind (\fun{\valone}{(\substcomp{\termtwo}{\varone}{\valone})^{(n)}})) \\
          &= \lub_{n<\omega} \termone^{(n)} \bind \lub_{n<\omega} (\fun{\valone}{(\substcomp{\termtwo}{\varone}{\valone})^{(n)}}) 
              \tag{Continuouity of $\bind$} \\
          &= \lub_{n<\omega} \termone^{(n)} \bind (\fun{\valone}{\lub_{n<\omega} (\substcomp{\termtwo}{\varone}{\valone})^{(n)}}) 
              \tag{Lub of functions} \\
          &= \sem{\termone} \bind (\fun{\valone}{\sem{\substcomp{\termtwo}{\varone}{\valone}}}).
      \end{align*}
    \item[\textbf{Case} 4.] We have:
      \begin{align*}
        \sem{\op(\termone_1, \hh, \termone_k)} 
          &= \lub_{n<\omega} (\op(\termone_1, \hh, \termone_k))^{(n+1)} \\
          &= \lub_{n<\omega} \opT(\termone_1^{(n)}, \hh, \termone_k^{(n)}) \\
          &= \opT(\lub_{n<\omega} \termone_1^{(n)}, \hh, \lub_{n<\omega} \termone_k^{(n)}) 
              \tag{Continuouity of $\opT$} \\
          &= \opT(\sem{\termone_1}, \hh, \sem{\termone_k}) 
      \end{align*}
  \end{varitemize}
\end{proof}
It is actually not hard to see that the function $\sem{\cdot}$ is the least
solution to the equations in Lemma \ref{semEquations}.

\section{On Relational Reasoning}

In this section we introduce the main machinery behind our soundness results.
The aim is to generalise notions and results from e.g.
\cite{Lassen/PhDThesis,Gordon/FOSSACS/01,Levy/ENTCS/2006} to take into account
generic effects. We will use results from the theory of coalgebras
\cite{Rutten/TCS/2000} to come up with a general notion of applicative
(bi)similarity parametric over a notion of observation, given through the
concept of relator.

\subsection{Relators}

The concept of relator \cite{Thijs/PhDThesis/1996,Levy/FOSSACS/2011} is an
abstraction meant to capture the possible ways a relation on a set $X$ can be
turned into a relation on $T X$. Recall that for an endofunctor $\functor :
\catone \to \catone$, an $\functor$-coalgebra \cite{Rutten/TCS/2000} consists
of an object $X$ of $\catone$ together with a morphism $\coalgX_X : X \to
\functor X$. As usual, we are just concerned with the case in which $\catone$
is $\set$.

\begin{definition}
  Let $\functor$ be an endofunctor on 
  $\set$, and $X,Y$ be sets. A relator $\relator$ for
  $\functor$ is a map that associates to each relation
  $\relone \subseteq X \times Y$ a relation 
  $\relator \relone \subseteq \functor X \times \functor Y$ 
  such that
  \begin{align*}
    =_{\functor X}                        &\subseteq  \Gamma(=_X)
                                            \tag{Rel-1} \label{rel-1}  \\
    \Gamma \relthree \circ \Gamma \relone &\subseteq \Gamma (\relthree \circ \relone) 
                                            \tag{Rel-2} \label{rel-2} \\
    \Gamma ((f \times g)^{-1} \relone)    &= (\functor f \times \functor g)^{-1} \relator \relone 
                                            \tag{Rel-3} \label{rel-3} \\
    \relone \subseteq \relthree           &\implies \Gamma \relone \subseteq \Gamma \relthree 
                                            \tag{Rel-4} \label{rel-4} 
  \end{align*}
  where for $f : Z \to X$, $g : W \to Y$ we have $(f \times g)^{-1}\relone =
  \{(z,w) \mid f(z)\ \relone\ g(w)\}$, and $=_X$ denotes the identity relation
  on $X$. A relator $\relator$ is \emph{conversive} if $\relator(\relone^c) =
  (\relator \relone)^c$, where $\relone^c$ denotes the converse of $\relone$.
\end{definition}

\begin{example}\label{relators}
  For each of the monads introduced in previous sections, we give
  some examples of relators. Most of these relators coincide 
  with the relation lifting of their associated functor. 
  It is in fact well known that for any weak-pullback 
  preserving functor, its relation lifting is a relator \cite{/Kurz/JLAMP/16}.
  We use the notation $\relatorSim$ for a relator
  aimed to capture the structure of a simulation relation, and
  $\relatorBisim$ for a relator aimed to capture the structure of a
  bisimulation relation. This distinction is not formal, 
  and only makes sense in the context of concrete examples: its purpose is 
  to stress that from formal view point, both concrete notions of similarity 
  and bisimilarity are modeled as forms of $\relator$-similarity (for a suitable 
  relator $\relator$).
  Let $\relone\subseteq X \times Y$:
  \begin{varitemize}
  \item 
    For the partiality monad $\monad X = X_\bot$ define the relators
    $\relatorSim_\bot, \relatorBisim_\bot$ by
    \begin{align*}
      u\ \relatorSim_{\bot}\relone\ v & \text{ iff } u = \inl(x) \implies v =
      \inl(y) \wedge x\ \relone\ y;   \\ u\ \relatorBisim_{\bot}\relone\ v
      &\text{ iff }           u = \inl(x)  \implies v = \inl(y) \wedge x\
      \relone\ y, \\ &\phantom{\text{ iff }} v = \inl(y) \implies u = \inl(x)
      \wedge x\ \relone\ y.
    \end{align*}
    Note that $u = \inl(x)$ means, in particular, $u \neq \inr(\bot)$. Thus,
    for instance, $u$ and $v$ are $\relator_\bot \relone$ related if whenever
    $u$ converges, so does $v$ and the values to which $u,v$ converge are
    $\relone$-related. The relator $\relatorBisim_\bot$ is conversive.
  \item 
    For the nondeterministic powerset monad $\powersetmonad$ define relators
    $\relatorSim_{\powersetmonad}$ and
    $\relatorBisim_{\powersetmonad}$ by
    \begin{align*}
      u\ \relatorSim_{\powersetmonad}\relone \ v  
      &\text{ iff } \forall x \in u.\ \exists y \in v.\ x\ \relone\ y;  \\
      u\ \relatorBisim_{\powersetmonad}\relone \ v  
      & \text{ iff } \forall x \in u.\ \exists y \in v.\ x\ \relone\ y, \\
      & \phantom{\text{ iff }} \forall y \in v.\ \exists x \in u.\ x\ \relone\ y.
    \end{align*}
    The relator $\relatorBisim_{\powersetmonad}$ is conversive.
  \item 
    For the probabilistic subdistributions monad $\distribution$ define relators
    $\relatorSim_{\distribution}$ and $\relatorBisim_{\distribution}$
    by
    \begin{align*}
      \distone\ \relatorSim_{\distribution}\relone \ \disttwo   
      & \text{ iff } \forall U \subseteq X.\ \distone(U) \leq \disttwo(\relone(U)); \\
      \distone\ \relatorBisim_{\distribution}\relone \ \disttwo   
      & \text{ iff }  \distone\ \relatorSim_{\distribution}\relone \ \disttwo\ \wedge\  
        \disttwo\ \relatorSim_{\distribution}\relone^c \ \distone;
    \end{align*}
    where $\relone(U) = \{y \in Y \mid \exists x \in U.\ x\ \relone\ y\}$ and 
    $\distone(U) = \sum_{x \in U} \distone(x)$. The relator $\relatorBisim_{\distribution}$ 
    is conversive.
  \item For the exception monad $\monad X = X + E$ define the relators 
    $\relatorSim_{\mathcal{E}}$ and 
    $\relatorBisim_{\mathcal{E}}$ by (letters $e, e'$ range over $E$)
    \begin{align*}
      u\ \relatorSim_{\mathcal{E}}\relone\ v   
        & \text{ iff }            u = \inr(e) \implies v = \inr(e') \wedge e = e', \\
        & \phantom{\text{ iff }}  u = \inl(x) \implies v = \inl(y) \wedge x\ \relone\ y;  \\
      u\ \relatorBisim_{\mathcal{E}}\relone\ v  
        & \text{ iff } u\ \relatorSim_{\mathcal{E}}\relone\ v,  \\
        & \text{\phantom{ iff }} v = \inr(e') \implies u = \inr(e) \wedge e = e', \\ 
        & \text{\phantom{ iff }} v = \inl(y) \implies u = \inl(x) \wedge x\ \relone\ y.
    \end{align*}
    The relator $\relatorBisim_{\mathcal{E}}$ is conversive.
  \item For the partiality and exception monad (i.e. the exception monad with
     divergence) $\monad X = (X + E)_\bot$ we can define relators simply
     composing relators for the partiality monad with relators for the
     exceptions monads (see Lemma \ref{algebra-of-relators}). Notably, define
     $\relatorSim_{\mathcal{E}_\bot}$ as $\relatorSim_{\bot} \circ
     \relatorSim_{\mathcal{E}}$ and $\relatorBisim_{\mathcal{E}_\bot}$ as
     $\relatorBisim_{\bot} \circ \relatorBisim_{\mathcal{E}}$. The relator
     $\relatorBisim_{\mathcal{E}_\bot}$ is conversive.
  \item For the state monad $\monad X = (X \times S)^S$ define the relator
    $\relatorBisim_{\mathcal{S}}$ by
    \begin{align*}
      f\ \relatorBisim_{\mathcal{S}}\relone \ g 
      & \text{ iff } \forall s \in S.\ s_1 = s_2 \text{ and } x_1\ \relone\ x_2, \\
      & \phantom{\text{ iff }} \text{where } (x_1, s_1) = f(s) \text{ and } (x_2, s_2) = g(s). 
    \end{align*}
    The relator $\relatorBisim_{\mathcal{S}}$ is conversive.
  \item For the output monad $\monad X = \stream{U} \times X_\bot$ we can define relators 
    based on the order defined in Example \ref{outputDoesNotWork}. 
    \begin{align*}
      (u, x)\ \relatorSim_{\mathcal{U}}\relone\ (w,y)   
        & \text{ iff }        (x = \inr(\bot) \wedge u \cpoleq w) \vee 
                              (x = \inl(x') \wedge y = \inl(y') \wedge x' \relone\ y'); \\
      (u,x)\ \relatorBisim_{\mathcal{U}}\relone\ (w,y)  
        & \text{ iff } (u,x)\ \relatorSim_{\mathcal{U}}\relone\ (w,y) \wedge 
                        (w,y)\ \relatorSim_{\mathcal{U}}\relone^c\ (u,x).
    \end{align*}
    The relator $\relatorBisim_{\mathcal{U}}$ is conversive.
  \end{varitemize}
\end{example}
Checking that the above are indeed relators is a tedious but easy exercise. It
is useful to know that the collection of relators is closed under certain
operations (see \cite{Levy/FOSSACS/2011} for proofs).
  \begin{lemma}[Algebra of Relators] \label{algebra-of-relators} 
    Let $\functor, G$ be endofunctors on $\set$. Then 
    \begin{varenumerate}
      \item Let $(\relator_i)_{i \in I}$ be a family of relators for
          $\functor$. The intersection $\bigcap_{i\in I} \relator_i$ defined
          by $(\bigcap_{i\in I} \relator_i)\relone = \bigcap_{i\in I}
          \relator_i(\relone)$ is a relator for $\functor$.
      \item The converse $\relator^c$ of $\relator$ defined by 
          $\relator^c(\relone) = (\relator\relone^c)^c$ is 
          a relator for $\functor$. We have the equality $(\relator^c)^c =
          \relator$ and, additionally, $\relator^c = \relator$ if $\relator$
          is conversive.
      \item  Let $\relator, \relator'$ be relators for $\functor, G$, respectively. Then 
          $\relator' \circ \relator$ is a relator for $G \circ \functor$.
          Moreover, if both $\relator$ and $\relator'$ are conversive, then so
          is $\relator' \circ \relator$.
      \item Given a relator $\relator$ for $\functor$, $\relator \cap
          \relator^c$ is the greatest (wrt the pointwise order) conversive
          relator for $\functor$ contained in $\relator$.
    \end{varenumerate}
\end{lemma}

We can now give a general notion of simulation with respect to a given relator.

\subsection{Bisimulation, in the Abstract}
A relator $\relator$ for a monad $\monad$ expresses the observable part of the
side-effects encoded by $\monad$. Its abstract nature allows to give abstract
definitions of simulation and bisimulation parametric in the notion of
observation given by $\relator$.
\begin{definition}\label{similarity}
  Let $\coalgX_X : X \to \functor X, \coalgX_Y : Y \to \functor Y$ be
  $\functor$-coalgebras:
  \begin{varenumerate}
    \item A $\relator$-simulation is a relation $\relone \subseteq X \times Y$
      such that $$x \relone y \implies \coalgX_X(x)\ \relator \relone\
      \coalgX_Y(y).$$
    \item $\relator$-similarity $\similar_{X,Y}^\relator$ is the largest 
      $\relator$-simulation.
  \end{varenumerate}
\end{definition}

\begin{example}
  It is immediate to see that the corresponding notions of
  $\relator$-similarity for the (bi)simulation relators of Example
  \ref{relators} coincide with widely used notions of (bi)similarity.
\end{example}
As usual, the notion of similiarity can be characterised coinductively as the
greatest fixed point of a suitable functional.

\begin{definition}
  Let $\coalgX_X: X \to \functor X$, $\coalgX_Y : Y \to \functor Y$ be
  $\functor$-coalgebras. Define the functional $\functional_{X,Y}^\relator :
  \rel{X}{Y} \to \rel{X}{Y}$ by $$\functional_{X,Y}^\relator (\relone) =
  (\coalgX_X \times \coalgX_Y)^{-1}(\relator \relone).$$
\end{definition}

When clear from the context, we will write $\functional_\relator$ and $\similar_\relator$ 
in place of $\functional_{X,Y}^\relator$ and $\similar_{X,Y}^\relator$.

\begin{lemma} The following hold:
  \begin{varenumerate}
    \item The functional $\functional_\relator$ is monotone, and thus has a
      greatest fixed point $\gfp{\functional_\relator}$.
    \item A relation $\relone$ is a $\relator$-simulation iff it is a post 
        fixed-point of $\functional_\relator$.
        Therefore, $\relator$-similarity coincides with
        $\gfp{\functional_\relator}$.
  \end{varenumerate}
\end{lemma}

\begin{proof}
  Monotonicity of $\functional_\relator$ directly follows from monotonicity of
  $\relator$, and thus it has greatest fixed point by Knaster-Tarski Theorem
  (recall that the set $\rel{X}{Y}$ carries a complete lattice structure under
  the inclusion order). A straightforward calculation shows that a relation
  $\relone$ is a $\relator$-simulation iff it is a post fixed-point of
  $\functional_\relator$. Together with point $1$, the latter implies
  $\gfp{\functional_\relator} = \similar_\relator$.
\end{proof}

\begin{proposition}\label{similarity-is-a-preorder}
  Let $\coalgX_X : X \to \functor X$ be an $\functor$-coalgebra.
  \begin{varenumerate}
    \item $\relator$-similarity is a preorder.
    \item If $\relator$ is conversive, then $\relator$-similarity is an equivalence relation.
  \end{varenumerate}
\end{proposition}

  \begin{proof}
    Let $\coalgX_X : X \to \functor X$ be an $\functor$-coalgebra.
    \begin{varenumerate}
      \item We prove that $\similar_\relator$ is reflexive by coinduction,
        showing that the identity relation $=_X$ on $X$ is a
        $\relator$-simulation. In fact, from $x =_X x$ we obtain $\coalgX_X(x)
        =_{\functor X} \coalgX_X(x)$, and thus we can conclude
        $\coalgX_{X}(x)\ \relator(=_X)\ \coalgX_{X}(x)$, by \eqref{rel-1}. \\
        We now show that $\similar_\relator$ is transitive. Suppose to have $x
        \similar_\relator y \similar_\relator z$. By very definition of
        $\similar_\relator$ there exist $\relator$-simulations $\relone,
        \relthree$ such that $x\ \relone\ y$ and $y\ \relthree\ z$, and thus
        $\coalgX_{X}(x)\ (\relator \relthree \circ \relator \relone)\
        \coalgX_{X}(z)$. Thanks to \eqref{rel-4} we can conclude that
        $\relthree \circ \relone$ is a $\relator$-simulation as well, meaning,
        in particular, that $x \similar_\relator z$
      \item We simply observe that if $\relone$ is a $\relator$-simulation,
        then so is $\conversive{\relone}$, for a conversive relator
        $\relator$.
    \end{varenumerate}
  \end{proof}

Since $\monad$ is a monad we consider relators that properly interact with the
monadic structure of $\monad$, which are also known as \emph{lax extensions}
for $\monad$~\cite{Barr/LMM/1970}.

\begin{definition}\label{laxExtension} 
  Let $\monad$ be a monad, $X,X',Y,Y'$ be
  sets, $f : X \to \monad X', g : Y \to \monad Y'$ be functions, and $\relone
  \subseteq X \times Y, \relthree \subseteq X' \times Y'$ be relations. We say
  that $\relator$ is a relator for $\monad$ if it is a relator for $\monad$
  regarded as a functor, and
  \begin{varitemize}
    \item $x\ \relone\ y \implies \uniT_X(x)\ \Gamma \relone\ \uniT_Y(y)$;
    \item $u\ \relator \relone\ v \implies (u \bind f)\ \relator \relthree\ (v \bind g)$, 
        whenever $x\ \relone\ y \implies f(x)\ \relator \relthree\ g(y)$.
  \end{varitemize}
\end{definition}

\begin{remark}\label{rem:laxExtension}
  Definition \ref{laxExtension} can be more compactly expressed using Kleisli star, thus requiring
  that
  \begin{align}
  &\relone \subseteq (\uniT_X \times \uniT_Y)^{-1}(\relator \relthree) \tag{Lax-Unit} \label{Lax-Unit} \\
  &\relone \subseteq (f, g)^{-1}(\relator \relthree) \Longrightarrow
  \relator \relone \subseteq (\kleisli{f} \times \kleisli{g})^{-1}(\relator \relthree) 
    \tag{Lax-Bind} \label{Lax-Bind}
  \end{align}
  or diagramatically 

  \begin{center}
  $
  \xymatrix{
      X       \ar[r]^{\relone}|{\scriptscriptstyle{/}}  
            \ar[d]_{\uniT_X}  &  
      Y     \ar[d]^{\uniT_Y}
      \\
      \monad X  \ar[r]_{\relator \relone}|{\scriptscriptstyle{/}}   &  
      \monad Y }
   $ \\
   $\vcenter{
    \xymatrix{
      X       \ar[r]^{\relone}|{\scriptscriptstyle{/}}  
            \ar[d]_{f}  &  
      Y     \ar[d]^{g}
      \\
      \monad X'   \ar[r]_{\relator \relthree}|{\scriptscriptstyle{/}}   &  
      \monad Y' } }
      \implies 
    \vcenter{
    \xymatrix{
      \monad X  \ar[r]^{\relator \relone}|{\scriptscriptstyle{/}}  
            \ar[d]_{\kleisli{f}}  &  
      \monad Y  \ar[d]^{\kleisli{g}}
      \\
      \monad X  \ar[r]_{\relator \relthree}|{\scriptscriptstyle{/}}   &  
      \monad Y }}
    $
    \end{center}
  
  where we write $\relone : X \nrightarrow Y$ for $\relone \subseteq X \times Y$.
\end{remark}

\begin{example}
   All relators of the form $\relator_{\monad}$ in Example \ref{relators} 
   are relators for $\monad$. Proving that is quite standard, with the exception of the probabilistic 
   case where the proof essentially relies on the Max Flow Min Cut Theorem \cite{Schrijver/Book/1986}. 
\end{example}

\begin{definition}\label{inductive}
  Let $\monad$ come with an \ocppo\ order $\cpoleq$.
  We say that $\relator \relone$ is \emph{inductive} 
  if for any $\omega$-chain $(u_n)_{n < \omega}$ in $\monad X$, we have:
  \begin{align}
    & \bot\ \relator \relone\ u \tag{$\omega$-comp 1} \label{omegaComp1} \\
    & (\forall n.\ u_n\ \relator \relone\ v)  \implies \lub_n u_n\ \relator \relone\ v.  \tag{$\omega$-comp 2} \label{omegaComp2}  
  \end{align}
  We say that $\relator$ \emph{respects} $\signature$ if 
  \begin{align}
     (\forall k.\ u_k\ \relator \relone\ v_k) \implies \op(u_1, \hh, u_n)\ \relator \relone\ \op(v_1, \hh, v_n)
    \tag{$\signature$-comp} \label{sigmaComp}
  \end{align}
  for any $\op \in \signature$, where $k \in \{1, \hh, \alpha(\op)\}$.
\end{definition}

\begin{remark}
  For a monad $\monad$ carrying a continuous $\signature$-algebra structure
  and a function $f: X \to \monad Y$, we required $\kleisli{f} : \monad X \to
  \monad Y$ to be continuous, $\monad X$ being an \ocppo. Since $\monad X$ is
  also a $\signature$-algebra, it seems natural to require $\kleisli{f}$ to be
  also a $\signature$-algebra homomorphism. In fact, such requirement implies
  condition (\ref{sigmaComp}) and has the advantage of being more general than
  the latter, not depending from the specific relator considered. Let $\monad$
  be a monad on $\set$. Following \cite{PlotkinPower/ACS/2003} we say that an
  $n$-ary algebraic operation (where $n$ is some set) associates to each set
  $X$ a function $\op_X : (\monad X)^n \to \monad X $ in such a way for every
  function $f : X \to TY$, the Kleisli extension $\kleisli{f}$ is a
  homomorphism. Recall that an $n$-ary generic effect is an element of $Tn$.
  As shown in \cite{PlotkinPower/ACS/2003}, there is a bijection from generic
  effects to algebraic operations as follows. Every $n$-ary generic effect $p$
  gives rise to an $n$-ary algebraic operation $\hat{p}$, where $\hat{p}_X$
  sends $u$ to $u^{\dagger} (p)$. Conversely, each $n$-ary algebraic operation
  $\op$ is $\hat{p}$ for a unique $n$-ary generic effect $p$, viz. $\sigma_n
  (\eta_n)$.

  We can now generalise our condition on $\monad$ by requiring it to be
  equipped with an $n$-ary algebraic operation for each $\op \in \signature$
  of arity $n$. This is the equivalent to extending our definitions by
  requiring the additional axiom that Kleisli extensions are homomorphisms.
  Moreover, requiring the bind operator to be strict in its first argument
  means that $\bot$ is an algebraic constant. Let us now prove that this
  condition implies condition ($\Sigma$-comp). For, suppose $\op$ has arity
  $n$, and $\forall k.\ u_k\ \relator \relone \ v_k$ holds, meaning that we
  have the square
    $$
     \xymatrix{
          n       \ar[r]^{=_n}|{\scriptscriptstyle{/}}  
                \ar[d]_{u}  &  
          n     \ar[d]^{v}
          \\
           TX   \ar[r]_{\relator \relone}|{\scriptscriptstyle{/}}   &  
           TY }
     $$ 
    As a consequence, we also have the square
    $$
     \xymatrix{
          Tn        \ar[r]^{\relator(=_n)}|{\scriptscriptstyle{/}}  
                  \ar[d]_{\kleisli{u}}  &  
          Tn      \ar[d]^{\kleisli{v}}
          \\
           TX   \ar[r]_{\relator \relone}|{\scriptscriptstyle{/}}   &  
           TY }
     $$ 
    and therefore
    $$
     \xymatrix{
          Tn        \ar[r]^{=_{Tn}}|{\scriptscriptstyle{/}}  
                  \ar[d]_{\kleisli{u}}  &  
          Tn      \ar[d]^{\kleisli{v}}
          \\
           TX   \ar[r]_{\relator \relone}|{\scriptscriptstyle{/}}   &  
           TY }
     $$ 
    Writing the algebraic operation associated with $\sigma$ as $\hat{p}$, 
    we have $p =_{T n} p$, and so  $u^{\dagger}(p)\ \relator \relone \ 
    v^{\dagger}(p)$, which essentially means
    $$\sigma_X (u_1,...,u_n)\ \relator \relone\ \sigma_Y(v_1,...,v_n).$$ 
    To the ends of this paper, condition (\ref{sigmaComp}) is sufficient and 
    thus we will use that throughout. 
\end{remark}

Following Abramsky \cite{Abramsky/RTFP/1990} we introduce Applicative Transition System 
(ATSs) over a monad (taking into account effectful computations) and define the notion 
of \emph{applicative simulation}. Let $\monad$ be a monad.

\begin{definition}\label{ApplicativeTransitionSystem}
  An applicative transition system (over $\monad$) consists of the following:
  \begin{varitemize}
    \item A state space made of a pair of sets $(X, Y)$ modelling closed terms 
      and values, respectively.
    \item An evaluation function $\varepsilon : X \to \monad Y$.
    \item An application function $\cdot : Y \to Y \to X$.
  \end{varitemize}
\end{definition}

The notion of ATS distinguishes between terms and values. As a consequence, we
often deal with pairs of relations $(\relone_X, \relone_Y)$, where $\relone_X,
\relone_Y$ are relations over $X$ and $Y$, respectively. We refer to such
pairs as $XY$-relations. $XY$-relations belongs to $2^{X \times X} \times 2^{Y
\times Y}$. The latter, being the product of complete lattices, is itself a
complete lattice.

\begin{definition}\label{applicative-similarity-ATS}
  Let $\relator$ be a relator for $\monad$. An applicative $\relator$-simulation is an $XY$-relation 
  $\relone = (\relone_X, \relone_Y)$ such that:
  \begin{align*}
  & x\ \relone_X x' \implies \varepsilon(x)\ \relator \relone_Y\ \varepsilon(x')     \tag{Sim-1}\label{Sim-1} \\
  & y\ \relone_Y y' \implies \forall w \in Y.\ y \cdot w \ \relone_X\ y' \cdot w.     \tag{Sim-2}\label{Sim-2}
  \end{align*}
\end{definition}

The above definition induces an operator 
$\behavior_\relator$ on $2^{X \times X} \times 2^{Y \times Y}$ 
defined for 
$\relone = (\relone_X, \relone_Y)$ as $(\behavior_\relator(\relone_X), \behavior_\relator(\relone_Y))$, where
\begin{align*}
  \behavior_\relator(\relone_X)  &= \{(x,x') \mid \varepsilon(x)\ \relator \relone_Y\ \varepsilon(x')\} \\
  \behavior_\relator(\relone_Y)  &= \{(y,y') \mid \forall w\in Y.\ y \cdot w\ \relone_X\ y'\cdot w\}.
\end{align*}
It is easy to prove that since $\relator$ is monotone, then so is $\behavior_\relator$. As a consequence, we 
can define applicative $\relator$-similarity as the greatest fixed point $\gfp{\behavior_\relator}$ 
of $\behavior_\relator$.

\begin{proposition}\label{applicativePreorder}
  The following hold:
  \begin{varenumerate}
    \item Applicative $\relator$-similarity $\similar_\relator$ is a preorder.
    \item If $\relator$ is conversive, then $\similar_\relator$ is an equivalence relation.
  \end{varenumerate}
\end{proposition} 

\begin{proof}
  The proof strictly follows the proof of Proposition \ref{similarity-
  is-a-preorder} (proving the desired properties with respect to clause
  \eqref{Sim-2} is straightforward). As an example, we show by coinduction
  that $\similar_\relator$ is reflexive by proving that the $XY$-identity
  relation $(=_X, =_Y)$ is an applicative $\relator$-simulation. From $x =_X
  x$ we infer $\varepsilon(x) =_{\monad Y} \varepsilon(x)$, and thus
  $\varepsilon(x)\ \relator(=_Y)\ \varepsilon(x)$, by \eqref{rel-1}. Moreover,
  we trivially have that $y =_Y y$ implies $y\cdot w =_X y \cdot w$.
\end{proof}

\section{Contextual Preorder and Applicative Similarity}
In the previous section, the axioms needed to generalise applicative
bisimilarity to our setting have been given. What remains to be done is to
appropriately instantiate all this to $\Lambda_\Sigma$.  We introduce the
notions of contextual preorder and applicative similarity (which will be then
extended to contextual equivalence and applicative \emph{bi}similarity). From
now we assume to have a monad $\monad$ carrying a continuous
$\signature$-algebra structure. Moreover, we assume any relator for $\monad$
to be inductive and to respect $\signature$. It is convenient to work with
generalisations of relations on closed terms (resp. values) called
$\lambda$-term relations.
\begin{definition}
  An open relation over terms is a set $\reloneTerms$ of triples $(\vars,
  \termone, \termtwo)$ where $\termone, \termtwo \in \terms(\vars)$.
  Similarly, an open relation over values is a set $\reloneValues$ of triples
  $(\vars, \valone, \valtwo)$ where $\valone, \valtwo \in \values(\vars)$. A
  $\lambda$-term relation is a pair $\relone = (\reloneTerms, \reloneValues)$
  made of an open relation $\reloneTerms$ over terms and an open relation
  $\reloneValues$ over values. A closed $\lambda$-term relation is a pair
  $\relone = (\reloneTerms, \reloneValues)$ where $\reloneTerms \subseteq
  \closedTerms \times \closedTerms$ and similarly for $\reloneValues$.
\end{definition}

\begin{remark}
  Formally, we can see an open relation over terms (and similarly over values)
  as an element of the cartesian product
  $\dependentProduct{\vars}{2^{\terms(\vars) \times \terms(\vars)}}$. That is,
  an open relation is a function that associates to each finite set $\vars$ of
  variables a (binary) relation between open terms in $\terms(\vars)$. Since,
  $2^{\terms(\vars) \times \terms(\vars)}$ is a complete lattice, for any
  finite set of variables $\vars$, then so is
  $\dependentProduct{\vars}{2^{\terms(\vars) \times \terms(\vars)}}$. That is,
  the set of open relations over terms (and over values) forms a complete
  lattice (the order is given pointwise). As a consequence, the set of
  $\lambda$-term relations is a complete lattice as well. These algebraic
  properties allow us to define open relations both inductively and
  coinductively, and, in particular, to extend notions and results developed
  in the relational calculus of \cite{Lassen/PhDThesis,Lassen/RelationalReasoning,
  Gordon/FOSSACS/01,Levy/ENTCS/2006}.
\end{remark}

We will use infix notation and write $\vars \imp \termone\ \reloneTerms\
\termtwo$ to indicate that $(\vars, \termone, \termtwo) \in \reloneTerms$. The
same convention applies to values and open relations over values. For a
$\lambda$-term relation $\relone = (\reloneTerms, \reloneValues)$, we often
write $\vars \imp \termone\ \relone\ \termtwo$ (i.e. $(\vars, \termone,
\termtwo) \in \relone$) for $\vars \imp \termone\ \reloneTerms\ \termtwo$
(i.e. $(\vars, \termone, \termtwo) \in \reloneTerms$). The same convention
holds for values and $\reloneValues$. Finally, we will use the notations
$\emptyset \imp \termone\ \relone\ \termtwo$ and $\termone\ \relone\ \termtwo$
interchangeably (and similarly for values).

There is a canonical way to extend a closed relation to an open one.

\begin{definition}
  Define the open extension operator mapping a closed relation over terms
  $\relone$ to the open relation $\open{\relone}$ (over terms) as follows:
  $(\vars, \termone, \termtwo) \in \open{\relone}$ iff $\termone, \termtwo \in
  \terms(\vars)$, and for all $\vals$, $\substcomp{\termone}{\vars}{\vals}\
  \relone\ \substcomp{\termtwo}{\vars}{\vals}$ holds.
\end{definition}  
The notion of open extension for a closed relation over values can be defined
in a similar way (using the appropriate notion of substitution).
 
The notion of reflexivity, symmetry and transitivity straightforwardly extends
to open $\lambda$-term relation (see e.g. \cite{Pitts/ATBC/2011}).

 \begin{definition}
  Let $\relone = (\reloneTerms, \reloneValues)$ be a $\lambda$-term relation.
  We say that $\relone$ is \emph{compatible} if the clauses in Figure
  \ref{compatibilityClauses} hold. We say that $\relone$ is a
  \emph{precongruence} if it is a compatible preorder. We say that $\relone$
  is a \emph{congruence} if it is a compatible equivalence.
\end{definition}

  \begin{figure}
   \begin{mdframed}
    \begin{center}
    $\vspace{-0.50cm}$
      \begin{align}
      & \forall \vars.\ \forall \varone \in \vars.\  \vars \imp \varone\ \reloneValues\ \varone   \tag{Comp1} \label{Comp1} \\
      & \forall \vars.\ \forall \varone \not \in \vars.\ \forall \termone, \termtwo.\ 
        \vars \cup \{\varone\} \imp \termone\ \reloneTerms\ \termtwo
        \implies \vars \imp \abs{\varone}{\termone}\ \reloneValues\ \abs{\varone}{\termtwo}   \tag{Comp2} \label{Comp2} \\
      & \forall \vars.\ \forall \valone, \valtwo\ \vars \imp \valone\ \reloneValues\ \valtwo 
        \implies \vars \imp \return{\valone}\ \reloneTerms\ \return{\valtwo}            \tag{Comp3} \label{Comp3} \\
      & \forall \vars.\ \forall \valone, \valone', \valtwo, \valtwo'.\ \vars \imp \valone\ \reloneValues\ \valone' \wedge 
        \vars \imp \valtwo\ \reloneValues\ \valtwo' \implies 
        \vars \imp \valone \valtwo\ \reloneTerms\ \valone' \valtwo'                 \tag{Comp4} \label{Comp4} \\
      & \forall \vars.\ \forall \varone \not \in \vars.\ \forall \termone, \termone', \termtwo, \termtwo'. \nonumber \\
        &  \vars \imp \termone\ \reloneTerms\ \termone' \wedge \vars \cup \{\varone\} \imp \termtwo\ \reloneTerms\ \termtwo'  
         \implies \vars \imp (\seq{\termone}{\termtwo})\ \reloneTerms\ (\seq{\termone'}{\termtwo'})   \tag{Comp5} \label{Comp5} \\
      & \forall \vars.\ \forall \op \in \signature.\ \forall \termone_1, \termtwo_1, \hh, \termone_n, \termtwo_n.\ \nonumber \\
      &  (\forall i \in \{1, \hh, n\}.\ \vars \imp \termone_i\ \reloneTerms\ \termtwo_i) \implies 
        \vars \imp \op(\termone_1, \hh, \termone_n)\ \reloneTerms\ \op(\termtwo_1, \hh, \termtwo_n)   \tag{Comp6} \label{Comp6}
  \end{align}
  \end{center}
  \end{mdframed}
  \caption{Compatibility Clauses.}
  \label{compatibilityClauses}
  \end{figure}

The following lemma will be useful.
    \begin{lemma}\label{decompositionComp}
      Let $\relone = (\reloneTerms, \reloneValues)$ be a $\lambda$-term relation. If $\relone$ is a preorder, then 
      properties \eqref{Comp4}, \eqref{Comp5}, \eqref{Comp6} are equivalent to their `unidirectional' versions:
      \begin{align}
      &   \forall \vars.\ \forall \valone, \valone', \valtwo.\ \vars \imp \valone\ \reloneValues\ \valone'  \implies 
        \vars \imp \valone \valtwo\ \reloneTerms\ \valone' \valtwo                
        \tag{Comp4L} \label{Comp4L} \\
      &   \forall \vars.\ \forall \valone, \valtwo, \valtwo'.\ \vars \imp \valtwo\ \reloneValues\ \valtwo'  \implies 
        \vars \imp \valone \valtwo\ \reloneTerms\ \valone \valtwo'                
        \tag{Comp4R} \label{Comp4R} \\
      &   \forall \vars.\ \forall \varone \not \in \vars.\ \forall \termone, \termone', \termtwo.\  \vars \imp \termone\ \reloneTerms\ \termone'   
        \implies \vars \imp (\seq{\termone}{\termtwo})\ \reloneTerms\ (\seq{\termone'}{\termtwo})   
        \tag{Comp5L} \label{Comp5L} \\
      &   \forall \vars.\ \forall \varone \not \in \vars.\ \forall \termone, \termtwo, \termtwo'.\  
        \vars \cup \{\varone\} \imp \termtwo\ \reloneTerms\ \termtwo'   
        \implies \vars \imp (\seq{\termone}{\termtwo})\ \reloneTerms\ (\seq{\termone}{\termtwo'})   
        \tag{Comp5R} \label{Comp5R}  \\
      & \forall \vars.\ \forall \op \in \signature.\ \forall \termone, \termtwo, \bar{\termone}, \bar{\termtwo}.\ 
         \vars \imp \termone\ \reloneTerms\ \termtwo \implies 
        \vars \imp \op(\bar{\termone}, \termone, \bar{\termtwo})\ \reloneTerms\ \op(\bar{\termone}, \termtwo, \bar{\termtwo})   
        \tag{Comp6C} \label{Comp6C}
      \end{align}
      where in \eqref{Comp6C} $\bar{\termone}, \bar{\termtwo}$ are possibly
      empty finite tuples of terms such that the sum of their lengths is equal
      to the ariety of $\op$ minus one.
    \end{lemma}
\begin{proof}
  The proof is straightforward. As a paradigmatic example, we show that clause
  \eqref{Comp5} is equivalent to the conjunction of clauses \eqref{Comp5L} and
  \eqref{Comp5R}. For the left to right implication, we assume that both
  \eqref{Comp5} and $\vars \imp \termone\ \reloneTerms\ \termone'$ hold, and
  show that $\vars \imp \seq{\termone}{\termtwo}\ \reloneTerms\
  \seq{\termone'}{\termtwo}$ holds as well, thus proving that \eqref{Comp5}
  implies \eqref{Comp5L} (the proof that \eqref{Comp5} implies \eqref{Comp5R}
  is morally the same). To prove the thesis, we observe that since $\relone$
  is reflexive, we have $\vars \cup \{\varone\} \imp \termtwo\ \reloneTerms\
  \termtwo$. Applying \eqref{Comp5} to the latter and $\vars \imp \termone\
  \reloneTerms\ \termone'$, we conclude $\vars \imp \seq{\termone}{\termtwo}\
  \reloneTerms\ \seq{\termone'}{\termtwo}$. \\ Now for the right to left
  direction. Assume \eqref{Comp5L} and \eqref{Comp5R} to be valid, and suppose
  both $\vars \imp \termone\ \reloneTerms\ \termone'$ and $\vars \cup
  \{\varone\} \imp \termtwo\ \reloneTerms\ \termtwo'$ to hold. We can apply
  \eqref{Comp5L} to the former, obtaining $\vars \imp
  \seq{\termone}{\termtwo}\ \reloneTerms\ \seq{\termone'}{\termtwo}$, and
  \eqref{Comp5R} to the latter, obtaining 
  $\vars \imp \seq{\termone'}{\termtwo}\ \reloneTerms\
  \seq{\termone'}{\termtwo'}$. The thesis now follows by transitivity of
  $\relone$.
\end{proof}

It is useful to characterise compatible relations via the notion of compatible refinement.

\begin{definition}
  Let $\relone = (\reloneTerms, \reloneValues)$ be a $\lambda$-term relation. 
  Define the compatible refinement $\ctxclosure{\relone}$ 
  of $\relone$ as the pair $(\ctxclosure{\relone}_\terms, \ctxclosure{\relone}_\values)$, where 
  $\reloneTerms$ and $\reloneValues$ are inductively defined by rules in Figure \ref{compatibleClosure}.
\end{definition}

    \begin{figure}[htbp]
    \begin{mdframed}
    \begin{center}
    \begin{tabular}{cc}
    \AXC{\phantom{$\{\varone\}$}} \RL{$\varone \in \vars$} \UIC{$\vars \imp \varone\ \ctxclosure{\relone}_\values\ \varone$} \DP 
    &
    \AXC{$\vars \cup \{\varone\} \imp \termone\ \reloneTerms\ \termtwo$} \RL{$\varone \not \in \vars$}
    \UIC{$\vars \imp \abs{\varone}{\termone}\ \ctxclosure{\relone}_\values\ \abs{\varone}{\termtwo}$}\DP \\ \\
    \AXC{$\vars \imp \valone\ \reloneValues\ \valtwo\phantom{\{ \}}$} 
    \UIC{$\vars \imp \return{\valone}\ \reloneTerms\ \return{\valtwo}$} \DP
    &
    \AXC{$\vars \imp \valone\ \reloneValues\ \valone'$}
    \AXC{$\vars \imp \valtwo\ \reloneValues\ \valtwo'$}
    \BIC{$\vars \imp \valone\valtwo\ \ctxclosure{\relone}_\terms\ \valone'\valtwo'$} \DP \\ \\
    \AXC{$\vars \imp \termone\ \reloneTerms\ \termone'$}
    \AXC{$\vars \cup \{\varone\} \imp \termtwo\ \reloneTerms\ \termtwo'$} \RL{$\varone \not \in \vars$}
    \BIC{$\vars \imp \seq{\termone}{\termtwo}\ \ctxclosure{\relone}_\terms\ \seq{\termone'}{\termtwo'}$}\DP
    &
    \AXC{$\vars \imp \termone_1\ \reloneTerms\ \termtwo_1 \quad \hh \quad \vars \imp \termone_n\ \reloneTerms\ \termtwo_n$}
    \UIC{$\vars \imp \op(\termone_1, \hh, \termone_n)\ \ctxclosure{\relone}_\terms\ \op(\termtwo_1, \hh, \termtwo_n)$} \DP
    \end{tabular}
    \end{center}
    \end{mdframed}
      \caption{Compatible Refinement Rules.}
      \label{compatibleClosure}
    \end{figure}

\begin{proposition}
  A $\lambda$-term relation $\relone$ is compatible iff 
  $\ctxclosure{\relone} \subseteq \relone$ holds.
\end{proposition}

The above notion of precongruence can be justified by observing that when a
relation $\relone$ is a preorder, being a precongruence does exactly mean to
be closed under the term constructors of the language. That could be formally
expressed by saying that $\relone$ is a precongruence if and only if $\vars
\imp \termone\ \relone\ \termtwo$ implies $\vars \imp \ctxone[\termone]\
\relone\ \ctxone[\termtwo]$, for any term context $\ctxone[\cdot]$. Defining
term contexts requires some care. In particular, when dealing with the
contextual preorder it is not possible to reason modulo $\alpha$-conversion,
thus making definition syntactically involved (see
\cite{Lassen/PhDThesis,Lassen/RelationalReasoning,Pitts/ATBC/2011} for
details). As remarked in \cite{Pitts/ATBC/2011}, it is possible to avoid those
difficulties by giving a coinductive characterisation of the contextual
preorder in the style of \cite{Lassen/PhDThesis,Gordon/FOSSACS/01}.
Essentially, the contextual preorder (and, similarly the contextual
equivalence) is defined as the largest compatible and preadequate (see
Definition \ref{adequateRelation}) $\lambda$-term relation. It is then easy to
provide a more syntactic definition of contextual preorder and to prove that
the two given definitions are equivalent
\cite{Gordon/FOSSACS/01,Lassen/PhDThesis,Pitts/ATBC/2011}.

The notion of adequacy defines the available observation on values. Being in
an untyped setting, it is customary not to observe them.

\begin{definition}\label{adequateRelation}
  Let $\relunit$ denote $\closedValues \times \closedValues$ seen as a closed relation, i.e. the trivial relation relating 
  all values. We say that a relation $\relone$ on terms is preadequate if 
  $$\emptyset \imp \termone\ \relone\ \termtwo \implies \sem{\termone}\ \relator \relunit\ \sem{\termtwo}$$
  where $\termone, \termtwo \in \closedTerms$.
  That is, a relation $\relone$ on terms is preadequate if whenever $\relone$ relates two closed terms, evaluating these programs 
  produces the same side-effects. 
  A $\lambda$-term relation $\relone = (\reloneTerms, \reloneValues)$ is preadequate iff $\reloneTerms$ is.
\end{definition}

\begin{example}
  It is easy to check that the above notion of adequacy (together with the relators in Example \ref{relators}) 
  captures standard notions of adequacy used for untyped $\lambda$-calculi.
  \begin{varitemize}
    \item Consider a calculus without operation symbols and with operational semantics over $(\closedValues)_\bot$. A relation 
      is preadequate if whenever $\emptyset \imp \termone\ \relone\ \termtwo$, then if $\termone$ converges, then so does 
      $\termtwo$.
    \item Consider a nondeterministic calculus with operational semantics over $\powersetmonad{\closedValues}$. A relation 
      is preadequate if whenever $\emptyset \imp \termone\ \relone\ \termtwo$, then if there exists a value $\valone$ to which 
      $\termone$ may converge (i.e. $\valone \in \sem{\termone}$), then there exists a value $\valtwo$ to which 
      $\termtwo$ may converge (i.e. $\valtwo \in \sem{\termtwo})$.
    \item Consider a probabilistic calculus with operational semantics over $\distribution{\closedValues}$. A relation 
      is preadequate if whenever $\emptyset \imp \termone\ \relone\ \termtwo$, then the probability of convergence of
      $\termone$ is smaller or equal than the probability of convergence of
      $\termtwo$.
  \end{varitemize}
\end{example}

Following \cite{Lassen/PhDThesis}, we shall define the \emph{$\relator$-contextual preorder} 
as the largest $\lambda$-term relation that is both compatible and preadequate.

\begin{definition}
  Let $\mathbb{C}\mathbb{A}$ be the set of relations on terms that are both 
  compatible and preadequate. Then define $\ctxpreord_\relator$ as $\bigcup \mathbb{C}\mathbb{A}$.
\end{definition}

\begin{proposition}
  The $\relator$-contextual preorder $\ctxpreord_\relator$ is a compatible and preadequate preorder.
\end{proposition}

\begin{proof}
  We prove that $\ctxpreord_\relator \in \mathbb{C}\mathbb{A}$. First of all
  note that $\mathbb{C}\mathbb{A}$ contains the open identity relation. In
  fact, the latter is clearly compatible. To see it is also preadequate
  suppose $\emptyset \imp \termone =_{\closedTerms} \termone$ so that
  $\sem{\termone} =_{\monad \values_0} \sem{\termone}$. By \eqref{rel-1}, we
  have $=_{\monad \closedValues} \subseteq \relator(=_{\closedValues})$.
  Moreover, by very definition of $\relunit$, we also have $=_{\values_0}
  \subseteq \relunit$ so that we can conclude $\sem{\termone}\ \relator
  \relunit\ \sem{\termone}$, by monotonicity of $\relator$. As a consequence,
  $\ctxpreord{\relator}$ satisfies
  \eqref{Comp1}. Observe also that $\eqref{Comp1}$ implies, in particular, 
  reflexivity of $\ctxpreord_\relator$. \\ We now show that
  $\ctxpreord_\relator$ satisfies
  \eqref{Comp2}. Suppose $(\vars \cup \{\varone\}, \termone, \termtwo) \in \ctxpreord_\relator$. That means 
  there exists a $\lambda$-term relation $\relone = (\reloneTerms,
  \reloneValues) \in \mathbb{C}\mathbb{A}$ such that $(\vars \cup \{\varone\},
  \termone, \termtwo) \in \reloneTerms$. Since $\relone$ is compatible it
  satisfies
  \eqref{Comp2}, and thus we have 
  $(\vars, \abs{\varone}{\termone}, \abs{\vartwo}{\termtwo}) \in
  \reloneValues$. It then follows $(\vars, \abs{\varone}{\termone},
  \abs{\varone}{\termtwo}) \in \ctxpreord_\relator$ (i.e. in its value
  component). Similarly, we can prove that $\ctxpreord_\relator$ satisfies
  \eqref{Comp3}.

  This approach does not work neither for \eqref{Comp4}, \eqref{Comp5} nor for
  \eqref{Comp6}. The reason is that all these clauses are multiple premises
  implications (and that badly interacts with the existential information
  obtained from being in $\ctxpreord_\relator$). Nonetheless, we can appeal to
  Lemma \ref{decompositionComp} to replace clauses \eqref{Comp4}-\eqref{Comp6}
  to single premiss implications (for which the proof works as for previous
  compatibility conditions). In order to use Lemma \ref{decompositionComp}, we
  need to prove that $\ctxpreord_\relator$ is transitive, and thus a preorder.
  For, it is sufficient to prove that $\mathbb{C}\mathbb{A}$ is closed under
  relation composition. The proof is rather standard and we just prove a
  couple of cases as examples.

  We first show that if $\relone = (\reloneTerms, \reloneTerms)$ and 
  $\relthree = (\relthreeTerms, \relthreeValues)$ are preadequate, then so is 
  $\relthree \circ \relone = (\relthreeTerms \circ \reloneTerms, \relthreeValues \circ \reloneValues)$. Suppose 
  $\emptyset \imp \termone\ \reloneTerms\ \termthree$ and $\emptyset \imp \termthree\ \relthreeTerms\ \termtwo$. Since 
  both $\relone$ and $\relthree$ are preadequate, we have $\sem{\termone}\ \relator \relunit\ \sem{\termthree}$ and 
  $\sem{\termthree}\ \relator \relunit\ \sem{\termtwo}$. By very definition of relator we have 
  $\relator \relunit \circ \relator \relunit \subseteq \relator(\relunit \circ \relunit)$. The latter is contained 
  in $\relator \relunit$, since $\relator$ is monotone and we trivially have $\relunit \circ \relunit$. 

  Proving that the composition of compatible relations is compatible is a straightforward exercise. 
  For instance, we show that if relations 
  $\relone = (\reloneTerms, \reloneValues), \relthree = (\relthreeTerms, \relthreeValues)$ satisfy 
  \eqref{Comp5}, then so does $\relthree \circ \relone$. For, suppose 
  $\vars \imp \termone\ (\relthreeTerms \circ \reloneTerms)\ \termone'$ and 
  $\vars \cup \{\varone\} \imp \termtwo\ (\relthreeTerms \circ \reloneTerms)\ \termtwo'$. As a consequence, we have 
  \begin{align}
    & \vars \imp \termone\ \reloneTerms\ \termone'' 
      \label{ctx-is-adequate-1}   \\
    & \vars \imp \termone''\  \relthreeTerms\ \termone'    
      \label{ctx-is-adequate-2}   \\
    & \vars \cup \{\varone\} \imp \termtwo\  \reloneTerms\ \termtwo''
      \label{ctx-is-adequate-3}  \\
    & \vars \cup \{\varone\} \imp \termtwo''\ \relthreeTerms\ \termtwo'.
      \label{ctx-is-adequate-4}
  \end{align}
  From \eqref{ctx-is-adequate-1} and \eqref{ctx-is-adequate-3} we infer 
  $\vars \imp \seq{\termone}{\termtwo}\ \reloneTerms\ \seq{\termone''}{\termtwo''}$, since 
  $\relone$ satisfies \eqref{Comp5}. Similarly, from \eqref{ctx-is-adequate-2} and 
  \eqref{ctx-is-adequate-4} we infer 
  $\vars \imp \seq{\termone''}{\termtwo''}\ \relthreeTerms\ \seq{\termone'}{\termtwo'}$. 
  We can conclude 
  $\vars \imp \seq{\termone}{\termtwo}\ (\relthreeTerms \circ \reloneTerms)\ \seq{\termone'}{\termtwo'}$.
\end{proof}

Finally, we define the notion of an \emph{applicative $\relator$-simulation}
observing that the collection of closed terms and values, together with
the operational semantics defined in previous section, carries an ATS
structure.

\begin{definition}\label{respectValues}
  A closed relation $\relone = (\reloneTerms, \reloneValues)$ respects values if for all
  closed values $\valone, \valtwo$, 
  $\valone\ \reloneValues\ \valtwo$ implies $\valone \valthree\ \reloneTerms\ \valtwo\valthree$,
  for any closed value $\valthree$.
\end{definition}

\begin{definition}\label{applicativeLTSofLambda}
  Define the ATS of closed $\lambda$-terms as follows:
  \begin{varitemize}
    \item The state space is given by the pair $(\closedTerms, \closedValues)$;
    \item The evaluation function is $\sem{\cdot} : \closedTerms \to \monad \closedValues$;
    \item The application function $\cdot : \closedValues \to \closedValues \to \closedTerms$ is 
      defined as term application: $\valone \cdot \valtwo = \valone \valtwo$.
  \end{varitemize}
\end{definition}
As a consequence, we can apply the general definition of applicative $\relator$-simulation to the 
ATS of $\lambda$-terms. Instantiating the general definition of applicative $\relator$-simulation we obtain:

\begin{definition}
  Let $\relator$ be a relator for the monad $\monad$. A closed relation $\relone = (\reloneTerms, \reloneValues)$ 
  is an applicative $\relator$-simulation if:
  \begin{varitemize}
    \item $\termone\ \reloneTerms\ \termtwo \implies \sem{\termone}\ \relator \reloneValues\ \sem{\termtwo}$;
    \item $\relone$ respects values. 
  \end{varitemize}
\end{definition}
We can then define applicative $\relator$-similarity $\similar_{\relator}$ as the 
largest applicative $\relator$-simulation, which we know to be 
a preorder by Proposition \ref{applicativePreorder}.
Most of the time the relator 
$\relator$ will be fixed; in those cases we will often write 
$\similar$ in place of $\similar_{\relator}$.

\begin{example}
  It is immediate to see that using the relators in Example \ref{relators} we recover well-known notions of 
  simulation and bisimulation. 
\end{example}
We want to prove that applicative similarity is a sound proof technique for contextual preorder. That is, 
we want to prove that $\similar_\relator\ \subseteq\ \ctxpreord_\relator$ holds. The 
relation $\ctxpreord_\relator$ being defined as the largest preadequate compatible relation,
the above inclusion is established by proving that 
$\similar_\relator$ is a precongruence. 

\section{Howe's Method and Its Soundness}
In this section we generalise Howe's technique to show that
applicative similarity is a precongruence, thus a sound proof
technique for the contextual preorder. Our generalisation shows how
Howe's method crucially (but only!) depends on the structure of the
monad modelling side-effects and the relators encoding their
associated notion of observation.
\begin{definition}\label{def:howe-def-1}
  Let $\relone$ be a closed $\lambda$-term relation. The Howe extension $\howe{\relone}$ of $\relone$ is defined as 
  the least relation 
  $\relthree$ such that $\relthree = \open{\relone} \circ \ctxclosure{\relthree}$. 
\end{definition}

It was observed in \cite{Levy/ENTCS/2006} that the above equation actually defines a unique relation.

\begin{lemma}
  Let $\relone$ be a closed $\lambda$-term relation. Then there is a unique relation 
  $\relthree$ such that $\relthree = \open{\relone} \circ \ctxclosure{\relthree}$.
\end{lemma}
As a consequence, $\howe{\relone}$ can be characterised both inductively and coinductively. 
Here we give two (well-known) equivalent inductive characterisations of $\howe{\relone}$. 

\begin{lemma}
  The following are equivalent and all define the relation $\howe{\relone}$.
  \begin{varenumerate}
    \item The Howe extension $\howe{\relone} = (\howe{\reloneTerms}, \howe{\reloneValues})$ of 
      $\relone$ is defined as the least relation closed under the following rules:
      \begin{center} 
        \AXC{$\vars \imp \termone\ \ctxclosure{\howe{\reloneTerms}}\ \termthree$}
        \AXC{$\vars \imp \termthree\ \open{\reloneTerms}\ \termtwo$}
        \BIC{$\vars \imp \termone\ \howe{\reloneTerms}\ \termtwo$}\DP $\quad$
        \AXC{$\vars \imp \valone\ \ctxclosure{\howe{\reloneValues}}\ \valthree$}
        \AXC{$\vars \imp \valthree\ \open{\reloneValues}\ \valtwo$}
        \BIC{$\vars \imp \valone\ \howe{\reloneValues}\ \valtwo$}\DP
      \end{center}
    \item The Howe extension $\howe{\relone}$ of $\relone$ is the relation inductively 
      defined by rules in Figure \ref{Howe-extension}.
  \end{varenumerate}
\end{lemma}

\begin{proof}
  It is easy to see that the functional $\functional$ on $\lambda$-term
  relations associated to Definition \ref{def:howe-def-1} (i.e. defined by
  $\functional(\relthree) = \open{\relone} \circ \ctxclosure{\relthree}$) is
  also the functional induced by rules in point $1$. We can prove by induction
  the equivalence between the relations defined in point $1$ and point $2$ (in
  fact, these are both defined inductively). This is tedious but easy, and
  thus the proof is omitted.
\end{proof}

  \begin{figure}
  \begin{mdframed}
  \begin{center}
  \begin{tabular}{c}
  \AXC{$\vars \imp \varone\ \open{\reloneValues}\ \valone$} \RL{$\howeVar$} 
  \UIC{$\vars \imp \varone\ \howe{\reloneValues}\ \valone$} \DP  
  \\ \\
  \AXC{$\vars \cup \{\varone\} \imp \termone\ \howe{\reloneTerms}\ \termthree$} 
  \AXC{$\vars \imp \abs{\varone}{\termthree}\ \open{\reloneValues}\ \valone$} \RL{$\howeAbs$}
  \BIC{$\vars \imp \abs{\varone}{\termone}\ \howe{\reloneValues}\ \valone$}\DP 
  \\ \\ 
  \AXC{$\vars \imp \valone\ \howe{\reloneValues}\ \valtwo$} 
  \AXC{$\vars \imp \return{\valtwo}\ \open{\reloneTerms}\ \termtwo$} \RL{$\howeRet$} 
  \BIC{$\vars \imp \return{\valone}\ \howe{\reloneTerms}\ \termtwo$} \DP 
  \\ \\
  \AXC{$\vars \imp \valone\ \howe{\reloneValues}\ \valone'$}
  \AXC{$\vars \imp \valtwo\ \howe{\reloneValues}\ \valtwo'$}
  \AXC{$\vars \imp \valone'\valtwo'\ \open{\reloneTerms}\ \termtwo$} \RL{$\howeApp$} 
  \TIC{$\vars \imp \valone \valtwo\ \howe{\reloneTerms}\ \termtwo$}\DP  
  \\ \\
  \AXC{$\vars \imp \termone\ \howe{\reloneTerms}\ \termthree$}
  \AXC{$\vars \cup \{\varone\} \imp \termone'\ \howe{\reloneTerms}\ \termthree'$} 
  \AXC{$\vars \imp \seq{\termthree}{\termthree'}\ \open{\reloneTerms}\ \termtwo$} \RL{$\howeSeq$}
  \TIC{$\vars \imp \seq{\termone}{\termone'}\ \howe{\reloneTerms}\ \termtwo$}\DP
  \\ \\
  \AXC{$\vars \imp \termone_k\ \howe{\reloneTerms}\ \termtwo_k\ (\forall k \geq n)$}
  \AXC{$\vars \imp \op(\termtwo_1, \hh, \termtwo_n)\ \open{\reloneTerms}\ \termtwo$} \RL{$\howeOp$}
  \BIC{$\vars \imp \op(\termone_1, \hh, \termone_n)\ \howe{\reloneTerms}\ \termtwo$} \DP
  \end{tabular}
  \end{center}
  \end{mdframed}
    \caption{Howe's Extension Rules.}
    \label{Howe-extension}
  \end{figure}

The following lemma states some nice properties of Howe's lifting of preorder
relations. The proof is standard and can be found in, e.g.,
\cite{DalLagoSangiorgiAlberti/POPL/2014}.

\begin{lemma}
  Let $\relone$ be a preorder. The following hold:
  \begin{varenumerate}
    \item $\relone \circ \howe{\relone} \subseteq \howe{\relone}$.
    \item $\howe{\relone}$ is compatible, and thus reflexive.
    \item $\relone \subseteq \howe{\relone}$.
  \end{varenumerate}
\end{lemma}

\begin{remark}
  To prove properties $2$ and $3$ it is actually sufficient to require
  $\relone$ to be reflexive, whereas property $1$, which we refer to
  transitivity of $\howe{\relone}$ wrt $\relone$, requires $\relone$ to be
  transitive. It is easy to see that a compatible relation is also reflexive.
\end{remark}

We now consider the Howe extension $\howe{\similar_\relator}$ of applicative
$\relator$-similarity. Since $\similar_\relator$ is a preorder (Proposition
\ref{applicativePreorder}), $\howe{\similar_\relator}$ is a compatible
relation containing $\similar_\relator$.

\begin{definition}
  A $\lambda$-term relation $\relone = (\reloneTerms, \reloneValues)$ is value-substitutive if 
  $\varone \imp \termone\ \reloneTerms\ \termtwo$ and $\emptyset \imp \valone\ \reloneValues\ \valtwo$ imply 
  $\emptyset \imp \substcomp{\termone}{\varone}{\valone}\ \reloneTerms\ \substcomp{\termtwo}{\varone}{\valtwo}$.
\end{definition}

\begin{lemma}\label{howeSubstitutive}
  The relation $\howe{\similar_\relator}$ is value-substitutive.
\end{lemma}

\begin{proof}
  The proof is standard, see e.g. \cite{DalLagoSangiorgiAlberti/POPL/2014}.
\end{proof}

Summing up, we have defined a compatible relation $\howe{\similar_\relator}$
which is value-substitutive and contains $\similar_\relator$. As a
consequence, to prove that the latter is compatible it is sufficient to prove
$\howe{\similar_\relator} \subseteq \similar_\relator$. We can proceed
coinductively, showing that $\howe{\similar_\relator}$ is an applicative
$\relator$-simulation. This is proved via the so-called Key Lemma. Before
proving the Key Lemma it is useful to spell out basic facts on the Howe
extension of applicative similarity that we will extensively use. In the
following we assume to have fixed a relator $\relator$, thus omitting
subscripts. Let $\relator$ be a relator.

\begin{lemma}\label{howeGammaTransitivity}
  The following hold:
  \begin{varenumerate}
    \item $\similar \circ \howesimilar\ \subseteq\   \howesimilar $.
    \item $(\Gamma \similar) \circ (\Gamma \howesimilar)   \subseteq \Gamma \howesimilar$.
  \end{varenumerate}
\end{lemma}

\begin{lemma}[Key Lemma]
  Let $\howe{\similar} = (\howe{\similarTerms}, \howe{\similarValues})$ be the Howe extension of 
  applicative similarity. 
  If $\emptyset \imp \termone \howe{\similarTerms} \termtwo$ and $\evalTo{\termone}{n}{\compone}$, then 
  $\compone\ \relator\howe{\similarValues}\ \sem{\termtwo}$.
\end{lemma}

\begin{proof}
  We proceed by induction on the derivation of the judgment $\evalTo{\termone}{n}{\compone}$.
  \begin{description}
    \item[\textbf{Case} $\ruleBot$.] Suppose to have $\evalTo{\termone}{0}{\bot}$. We are done since 
      $\relator$ is inductive, and thus $\bot\ \relator\howe{\similarValues}\ \sem{\termtwo}$ trivially holds 
      (see property \eqref{omegaComp2}).
    \item[Case $\ruleReturn$.] Suppose to have $\evalTo{\return{\valone}}{n+1}{\unit{\valone}}$. 
      By hypothesis we have 
      $\emptyset \imp \return{\valone} \howe{\similarTerms} \termtwo$, so that the latter must have been obtained 
      as the conclusion of an instance of rule $\howeRet$. As a consequence, we have 
      $\emptyset \imp \valone \howe{\similarValues} \valtwo$ and $\return{\valtwo} \similarTerms \termtwo$, for 
      some value $\valtwo$. 
      We can now appeal to \eqref{Lax-Unit}, thus inferring 
      $\unit{\valone}\ \relator\howe{\similarValues}\ \unit{\valtwo}$ from 
      $\emptyset \imp \valone \howe{\similarValues} \valtwo$.
      By very definition of applicative similarity, $\return{\valtwo} \similarTerms \termtwo$ implies 
      $\sem{\return{\valtwo}}\ \relator\similarValues\ \sem{\termtwo}$, 
      i.e. $\unit{\valtwo}\ \relator\similarValues\ \sem{\termtwo}$. Therefore, we have 
      $\unit{\valone}\ (\relator\similarValues) \circ (\relator\howe{\similarValues})\ \sem{\termtwo}$, 
      from which the thesis follows by Lemma \ref{howeGammaTransitivity}.
    \item[\textbf{Case} $\ruleApp$.] Suppose the judgment 
      $\evalTo{(\abs{\varone}{\termone})\valone}{n+1}{\compone}$ has been obtained 
      from the judgment $\evalTo{\substcomp{\termone}{\varone}{\valone}}{n}{\compone}$. By   
      hypothesis we have $\emptyset \imp (\abs{\varone}{\termone})\valone \howe{\similarTerms} \termtwo$, meaning that the 
      latter must have been obtained as the conclusion of an instance of $\howeApp$. We thus obtain 
      $\emptyset \imp \abs{\varone}{\termone} \howe{\similarValues} \valtwo$,
      $\emptyset \imp \valone \howe{\similarValues} \valthree$ and 
      $\valtwo \valthree \similarTerms \termtwo$,
      for values $\valtwo, \valthree$. Looking at the first of these three judgments, we see that it 
      must be the conclusion of an instance of rule $\howeAbs$. Therefore, we have
      $\{\varone\} \imp \termone \howe{\similarTerms} \termthree$ and 
      $\abs{\varone}{\termthree} \similarValues \valtwo$.
      Since $\howe{\similar}$ is value-substitutive, from 
      $\{\varone\} \imp \termone \howe{\similarTerms} \termthree$ and
      $\emptyset \imp \valone \howe{\similarValues} \valthree$  
      we conclude  
      $\emptyset \imp \substcomp{\termone}{\varone}{\valone} \howe{\similarTerms} \substcomp{\termthree}{\varone}{\valthree}$.
      We can now apply the induction hypothesis on the latter and $\evalTo{\substcomp{\termone}{\varone}{\valone}}{n}{\compone}$, 
      obtaining $\compone\ \relator(\howe{\similarValues})\ \sem{\substcomp{\termthree}{\varone}{\valthree}}$.
      Since $\similar$ respects values, from 
      $\abs{\varone}{\termthree} \similarValues \valtwo$ we infer 
      $(\abs{\varone}{\termthree})\valthree \similarTerms \valtwo \valthree$, which gives, by very definition of applicative 
      similarity, $\sem{(\abs{\varone}{\termthree})\valthree}\ \relator(\similarValues)\ \sem{\valtwo \valthree}$. By Lemma 
      \ref{semEquations}, $\sem{(\abs{\varone}{\termthree})\valthree} = \sem{\substcomp{\termthree}{\varone}{\valthree}}$, 
      and thus, $\compone\ \relator(\howe{\similarValues})\ \sem{\valtwo \valthree}$, by Lemma \ref{howeGammaTransitivity}. 
      Finally, from $\valtwo \valthree \similarValues \termtwo$ we obtain 
      $\sem{\valtwo \valthree}\ \relator(\similarValues)\ \sem{\termtwo}$, which allows us to conclude 
      $\compone\ \relator(\similarValues)\ \sem{\termtwo}$ by Lemma \ref{howeGammaTransitivity}.
    \item[\textbf{Case} $\ruleSeq$.] Suppose the judgment 
      $\evalTo{(\seq{\termone}{\termone'})}{n+1}{\compone \bind (\fun{\valone}{\comptwo_{\valone}})}$ has been obtained from 
      $\evalTo{\termone}{n}{\compone}$ and $\evalTo{\substcomp{\termone'}{\varone}{\valone}}{n}{\comptwo_{\valone}}$. By 
      hypothesis we have $\emptyset \imp \seq{\termone}{\termone'} \howe{\similarTerms} \termtwo$, which must have been obtained 
      via an instance of rule $\howeSeq$ thus giving  
      $\emptyset \imp \termone \howe{\similarTerms} \termthree$, $\{\varone\} \imp \termone' \howe{\similarTerms} \termthree'$ 
      and $\emptyset \imp \seq{\termthree}{\termthree'}\ \similarTerms\ \termtwo$. We can apply the induction hypothesis on 
      $\evalTo{\termone}{n}{\compone}$ and $\emptyset \imp \termone \howe{\similarTerms} \termthree$ obtaining 
      $\compone\ \relator(\howe{\similarValues})\ \sem{\termthree}$. We now claim to have 
      $$\compone \bind (\fun{\valone}{\comptwo_{\valone}})\ \relator(\howe{\similarValues})\ 
      \sem{\termthree} \bind (\fun{\valone}{\sem{\substcomp{\termthree'}{\varone}{\valone}}}).$$
      The latter is equal to $\sem{\seq{\termthree}{\termthree'}}$, by Lemma \ref{semEquations}. Besides, 
      $\emptyset \imp \seq{\termthree}{\termthree'} \similarTerms \termtwo$ entails 
      $\sem{\seq{\termthree}{\termthree'}}\ \relator(\similarValues)\ \sem{\termtwo}$: we conclude 
      $\compone \bind (\fun{\valone}{\comptwo_{\valone}})\ \relator(\howe{\similarValues})\ \sem{\termtwo}$, 
      by Lemma \ref{howeGammaTransitivity}. \\
      The above claim directly follows from \eqref{Lax-Bind}. In fact, since 
      $\compone\ \relator(\howe{\similarValues})\ \sem{\termthree}$ holds,
      by \eqref{Lax-Bind} it is sufficient to prove that $\valone \howe{\similarValues} \valtwo$ implies 
      $\comptwo_{\valone}\ \relator(\howe{\similarValues})\ \sem{\substcomp{\termthree'}{\varone}{\valtwo}}$. Assume 
       $\valone \howe{\similarValues} \valtwo$, i.e. $\emptyset \imp \valone \howe{\similarValues} \valtwo$. 
      The latter, together with $\{\varone\} \imp \termone' \howe{\similarTerms} \termthree'$, implies 
      $\emptyset \imp \substcomp{\termone'}{\varone}{\valone} \howe{\similarTerms} \substcomp{\termthree'}{\varone}{\valtwo}$, 
      since $\howe{\similar}$ is value-substitutive. We can finally apply the inductive hypothesis on the latter and 
      $\evalTo{\substcomp{\termone'}{\varone}{\valone}}{n}{\comptwo_{\valone}}$, thus concluding the wanted thesis.
    \item[\textbf{Case} $\ruleOp$.] Suppose the judgment 
      $\evalTo{\op(\termone_1, \hh, \termone_k)}{n+1}{\opT(\compone_1, \hh, \compone_k)}$ 
      has been obtained from $\evalTo{\termone_1}{n}{\compone_1}, \hh, \evalTo{\termone_k}{n}{\compone_k}$. 
      By hypothesis we have 
      $\emptyset \imp \op(\termone_1, \hh, \termone_k) \howe{\similarTerms} \termtwo$, 
      which must be the conclusion of an instance of rule $\howeOp$. As a consequence, judgments 
      $\emptyset \imp \termone_1 \howe{\similarTerms} \termtwo_1, \hh, \emptyset \imp \termone_k \howe{\similarTerms} \termtwo_k$ 
      and $\emptyset \imp \op(\termtwo_1, \hh, \termtwo_k) \similarTerms \termtwo$ hold, 
      for some terms $\termtwo_1, \hh, \termtwo_k$. 
      We can repeatedly apply the induction hypothesis on $\evalTo{\termone_i}{n}{\compone_i}$ and 
      $\emptyset \imp \termone_i \howe{\similarTerms} \termtwo_i$, 
      inferring $\compone_i\ \relator\howe{\similarValues}\ \sem{\termtwo_i}$, 
      for all $i \in \{1, \hh, k\}$. \eqref{sigmaComp} allows to conclude 
      $\opT(\compone_1, \hh, \compone_k)\ \relator\howe{\similarValues}\ \opT(\sem{\termtwo_1}, \hh, \sem{\termtwo_k})$. 
      By Lemma \ref{semEquations} the latter is equal to $\sem{\op(\termtwo_1, \hh, \termtwo_k)}$. Finally, from 
      $\emptyset \imp \op(\termtwo_1, \hh, \termtwo_k) \similarTerms \termtwo$ we infer 
      $\sem{\op(\termtwo_1, \hh, \termtwo_k)}\ \relator\similarValues\ \sem{\termtwo}$ from which the thesis follows by Lemma 
      \ref{howeGammaTransitivity}.
  \end{description}
\end{proof}

\begin{corollary}
  The relation $\howe{\similar_\relator}$ is an applicative $\relator$-simulation.
\end{corollary}
\begin{proof}
  Suppose $\termone \howe{\similarTerms} \termtwo$. We have to prove 
  $\sem{\termone}\ \relator\howe{\similarValues}\ \sem{\termtwo}$, i.e. 
  $\lub_{\evalTo{\termone}{n}{\compone}} \compone\ \relator(\howe{\similarValues})\ \sem{\termtwo}$. The latter follows 
  from \eqref{omegaComp1} by the Key Lemma. Finally, since $\howe{\similar}$ is compatible, it clearly respects values. 
\end{proof}

\begin{theorem}\label{similarity-is-a-precongruence}
  Similarity is a precongruence. Moreover, it is sound for contextual preorder $\ctxpreord_\relator$.
\end{theorem}
\begin{proof}
  We already know $\similar_\relator$ is a preorder. By previous corollary it follows that $\similar_\relator$ 
  coincides with $\howe{\similar_\relator}$, so that $\similar_\relator$ is also compatible, and thus a 
  precongruence. Now for soundness. We have to prove $\similar_\relator\ \subseteq\ \ctxpreord_\relator$. 
  Since $\ctxpreord_\relator$ is defined as the largest preadequate compatible relation,
  it is sufficient to prove that $\similar_\relator$ 
  is preadequate (we have already showed it is compatible), which directly follows from \ref{Sim-1}, since 
  $\similarValues \subseteq \relunit$.
\end{proof}

\section{Bisimilarity, Two-similarity and Contextual Equivalence}
In this section we extend previous definitions and results to come up with
sound proof techniques for contextual \emph{equivalence}. In particular, by
observing that contextual equivalence always coincides with the intersection
between the contextual preorder and its converse, Theorem \ref{similarity-
is-a-precongruence} implies that two-way similarity (i.e. the intersection
between applicative similarity and its converse) is contained in contextual
equivalence. Applicative bisimilarity being finer than two-way similarity, we
can also conclude the former to be a sound proof technique for contextual
equivalence.

Given a relator $\relator$, we can extract a canonical notion of
$\relator$-bisimulation from the one of $\relator$-simulation following the
idea that a bisimulation is a relation $\relone$ such that both $\relone$ and
$\relone^c$ are simulations. Recall that given a relator $\relator$ we can
define a converse operation $\relator^c$ as $\relator^c(\relone) =
(\relator(\relone^c))^c$. $\relator^c$ is indeed a relator. Similarly, we have
proved that the intersection of relators is again a relator.

\begin{definition}[$\relator$-bisimulation]\label{gamma-bisimulation}
  Given a relator $\relator$, we say that a relation $\relone$ is a
  $\relator$-bisimulation if it is a $(\relator \cap \relator^c)$-simulation.
\end{definition}

\begin{proposition}
  Let $\relator$ be a relator. A relation $\relone$ is a $\relator$-bisimulation 
  if and only if both $\relone$ and $\relone^c$ are $\relator$-simulation.
\end{proposition}

Since, by Lemma \ref{algebra-of-relators}, $\relator \cap \relator^c$ is a relator, 
we can define $\relator$-bisimilarity 
$\bisim_\relator$ as $(\relator \cap \relator^c)$-similarity. 

\begin{lemma}
  Let $\relator$ be a relator. $\relator$-bisimilarity is an equivalence relation.
\end{lemma}
\begin{proof}
  From Lemma \ref{applicativePreorder} we know that $\bisim_\relator$ is a preorder, 
  whereas Lemma \ref{algebra-of-relators} shows that $\relator \cap \relator^c$ is 
  conversive. We conclude $\bisim_\relator$ to be an equivalence relation.
\end{proof}

\begin{definition}
  Let $\relator$ be a relator. Define $\relator$-cosimilarity $\cosimilar_\relator$ as $(\similar_\relator)^c$. 
  Define $\relator$ two-way similarity $\simeq_\relator$ as $\similar_\relator \cap \cosimilar_\relator$. 
\end{definition}
As usual, bisimilarity is finer than two-way similarity, meaning that
$\bisim_\relator\ \subseteq\ \simeq_\relator$. Moreover, taking $\relator$ to
be the simulation relator for the powerset monad (see Example \ref{relators}),
we have that $\bisim_\relator$ and $\simeq_\relator$ do not coincide. See e.g.
\cite{Lassen/PhDThesis,Pitts/ATBC/2011}.

\begin{proposition}
   Let $\relator$ be a relator. Then, $\bisim_\relator\ \subseteq\ \simeq_\relator$, and the inclusion is, in general,
    strict.
\end{proposition}

Recall that we have defined the $\relator$-contextual preorder
$\ctxpreord_\relator$ as the largest relation that is both compatible and
$\relator$-preadequate. In analogy with what we did for simulation and
bisimulation we can give the following:

\begin{definition}
  Let $\relator$ be a relator. Define $\relator$-contextual equivalence $\ctxequiv_\relator$ as the 
  largest relation that is both compatible and $(\relator \cap \relator^c)$-preadequate. That is, define 
  $\ctxequiv_\relator$ as $\ctxpreord_{\relator \cap \relator^c}$.
\end{definition}

\begin{lemma}\label{cocontextual-coinduction}
  Let $\relator$ be a relator. The cocontextual preorder $\geq_\relator$ is the largest relation that is both 
  $\relator^c$-preadequate and compatible.
\end{lemma}

\begin{proof}
  First of all observe that if a relation $\relone$ is $\relator$-preadequate,
  then $\relone^c$ is $\relator^c$-preadequate. For, suppose $\termtwo\
  \relone^c\ \termone$, so that $\termone\ \relone\ \termtwo$. Since $\relone$
  is $\relator$-preadequate, we have $\sem{\termone}\ \relator\relunit\
  \sem{\termtwo}$, and thus $\sem{\termtwo}\ (\relator \relunit)^c\
  \sem{\termone}$. From $\relunit^c = \relunit$ we can conclude
  $\sem{\termtwo}\ \relator^c(\relunit)\ \sem{\termone}$.

  As a consequence, since $\ctxpreord_\relator$ is $\relator$-preadequate, we
  have that $\geq_\relator$ is $\relator^c$-preadequate. Moreover,
  compatibility of $\ctxpreord_\relator$ implies compatibility of
  $\geq_\relator$. It remains to prove that $\geq_\relator$ is the largest
  $\relator^c$-preadequate and compatible relation. Let $\relone$ be a
  $\relator^c$-preadequate and compatible relation. We show $\relone\
  \subseteq\ \geq_\relator$ by showing $\relone^c\ \subseteq\
  (\geq_\relator)^c$, i.e. $\relone^c\ \subseteq\ \ctxpreord_\relator$. We
  proceed by coinduction showing that $\relone^c$ is $\relator$-preadequate
  and compatible. Compatibility of $\relone^c$ directly follows from that of
  $\relone$. Moreover, since $\relone$ is $\relator^c$-preadequate,
  $\relone^c$ is $(\relator^c)^c$-preadequate. A simple calculation shows that
  $(\relator^c)^c = \relator$, so that we are done.
\end{proof}

Although bisimilarity is finer than two-way similarity, this is not
the case for contextual equivalence and the associated contextual
preorders.
\begin{proposition}
  Let $\relator$ be a relator. Then, $\ctxequiv_\relator\ =\ \ctxpreord_\relator \cap \geq_\relator$. 
\end{proposition}

\begin{proof}
  First of all observe that since $\relunit = \relunit^c$, a relation $\relone$ is $(\relator \cap \relator^c)$-preadequate 
  if $\termone\ \reloneTerms\ \termtwo$ implies that both $\sem{\termone}\ \relator \relunit\ \sem{\termtwo}$ and 
  $\sem{\termtwo}\ \relator\relunit\ \sem{\termone}$ hold.
  Since $\ctxequiv_\relator$ is defined coinductively, to prove that it contains $\ctxpreord_\relator \cap \geq_\relator$ 
  it is sufficient to prove that $\ctxpreord_\relator \cap \geq_\relator$ is compatible and 
  $(\relator \cap \relator^c)$-preadequate. Standard calculations show that the set of compatible relations is closed under 
  converse and intersection. Since $\ctxpreord_\relator$ is compatible, then so is $\geq_\relator$ and thus 
  $\ctxpreord_\relator \cap \geq_\relator$. We show that $\ctxpreord_\relator \cap \geq_\relator$ is
  $(\relator \cap \relator^c)$-preadequate. Suppose $\termone\ (\ctxpreord_\relator \cap \geq_\relator)\ \termtwo$, so that 
  both $\termone \ctxpreord_\relator \termtwo$ and $\termone \geq_\relator \termtwo$ hold. From the former it follows 
  $\sem{\termone}\ \relator \relunit\ \sem{\termtwo}$, whereas from the latter we infer $\termtwo \ctxpreord_\relator \termone$
  and thus $\sem{\termtwo}\ \relator \relunit\ \sem{\termone}$.

  We now show that $\ctxequiv_\relator$ is contained in $\ctxpreord_\relator \cap \geq_\relator$. Since $\ctxpreord_\relator$ 
  is defined coinductively, to prove $\ctxequiv_\relator\ \subseteq\ \ctxpreord_\relator$ it is sufficient to prove that 
  $\ctxequiv_\relator$ is compatible and $\relator$-preadequate, which is indeed the case. Thanks to Lemma \ref{cocontextual-coinduction} 
  we can proceed coinductively to prove $\ctxequiv_\relator\ \subseteq\ \geq_\relator$ as well. In fact, it is sufficient to 
  prove that $\ctxequiv_\relator$ is $\relator^c$-preadequate, which is trivially the case.
\end{proof}

We can finally prove our soundness result. 
\begin{theorem}[Soundness]\label{soundness} 
  Let $\relator$ be a relator. Two-way similarity $\simeq_\relator$ is a congruence, and thus sound for contextual 
  equivalence $\ctxequiv_\relator$. Since bisimilarity $\sim_\relator$ is finer than $\simeq_\relator$, it is sound for 
  $\ctxequiv_\relator$ as well.
\end{theorem}

\begin{proof}
  From Theorem \ref{similarity-is-a-precongruence} we know that $\similar_\relator$ is a precongruence and that 
  $\similar_\relator\ \subseteq\ \ctxpreord_\relator$. It follows $\cosimilar_\relator$ is a precongruence as well, 
  and that $\cosimilar_\relator\ \subseteq\ \geq_\relator$ holds. We can conclude $\simeq_\relator$ is a congruence 
  and $\simeq_\relator\ \subseteq\ \ctxequiv_\relator$. Since $\sim_\relator\ \subseteq\ \simeq_\relator$, we also 
  have $\sim_\relator\ \subseteq\ \ctxequiv_\relator$.
\end{proof}

Noticeably, Theorem~\ref{soundness} can be seen as a proof of
soundness for applicative bisimilarity in any calculus
$\Lambda_\Sigma$ which respects our requirements (see
Definition~\ref{laxExtension},~\ref{inductive}), and in particular for those described
in Example \ref{relators}.  The case of probabilistic calculi is
illuminating: the apparent complexity of all proofs of congruence from
the literature~\cite{DalLagoSangiorgiAlberti/POPL/2014,CrubilleDalLago/ESOP/2014}
has been confined to the proof that the relator for subdistributions
satisfies our axioms.

We can rely on Theorem \ref{soundness} to prove that the terms 
$W^{\mathsf{raise}}$ and $Z^{\mathsf{raise}}$, our example programs from Section
\ref{sect:informal}, being bisimilar, are indeed contextually equivalent.
This only requires checking
that the map $\relator_\distribution \circ \relator_{\mathcal{E}}$
(see Example \ref{relators}) is an inductive relator for the monad
$\monad X = \distribution(X + E)$ (which trivially carries a
continuous $\signature$-algebra structure) respecting operations in
$\signature$. This is an easy exercise, and does not require any 
probabilistic reasoning.

Let $(\distribution, \delta, (\cdot)^{\distribution})$ denote the
subdistributions monad, where we write $f^{\distribution}$ for the Kleisli
lifting of $f$ and $\delta_X$ for the Dirac distribution on the set $X$.
Similarly, let $(\exception, \epsilon, (\cdot)^{\exception})$ denote the
exception monad, where we write $f^{\exception}$ for the Kleisli lifting of
$f$, and $\epsilon$ for unit of $\exception$ (see Example \ref{monadExample}
for formal definitions). Moreover, recall that we have relators
$\relator_{\distribution}$ and $\relator_{\exception}$ for $\distribution$ and
$\exception$, respectively (see Example \ref{relators}). A standard
calculation shows that we have the following:
\begin{proposition}
  The functor $\distribution \circ \exception$ induces a Kleisli triple 
  $(\distribution \circ \exception, \uniT, \kleisli{(\cdot)})$, where 
  the unit $\uniT$ is defined, for any set $X$, by $\uniT_X = \delta_{\exception(X)} \circ \epsilon_X$, 
  whereas for a function $f : X \to \distribution \exception (Y)$ 
  the Kleisli extension $\kleisli{f}$ of $f$ is defined as $(f_*)^{\distribution}$, 
  where $f_* : \exception(X) \to \distribution \exception(Y)$ is defined by
    $$f_*(u) = \begin{cases}
                f(x)                      & \text{if } u = \inl(x); \\
                \delta_{\exception(Y)}(u) & \text{otherwise.} 

              \end{cases}
    $$
\end{proposition}

Being defined as composition of relators, the map 
$\relator_{\mathcal{D}} \circ \relator_{\mathcal{E}}$ 
(also written $\relator_{\mathcal{D}} \relator_{\mathcal{E}}$) 
is a relator for the functor $\distribution \circ \exception$. 
We show that it also satisfies conditions \eqref{Lax-Unit} and 
\eqref{Lax-Bind}, meaning that it is a relator for 
$\distribution \circ \exception$, regarded as a monad. 
In order to have a more readable proof we use a couple of 
simple auxiliary lemmas.
In the following, let $f : X \to \distribution \exception Z$ and 
$g : Y \to \distribution \exception W$ be maps, and 
$\relone \subseteq X \times Y$, $\relthree \subseteq Z \times W$ be relations.

\begin{lemma}\label{laxProbEXcepHelp}
  The following implication holds
  $$\relone \subseteq (f \times g)^{-1}(\relator_{\mathcal{D}} \relator_{\mathcal{E}} \relthree) 
    \implies \relator_{\mathcal{E}} \relone \subseteq (f_* \times g_*)^{-1}(\relator_{\mathcal{D}} \relator_{\mathcal{E}} \relthree)
  $$
\end{lemma}
\begin{proof}
  Suppose $\relone \subseteq (f \times g)^{-1}(\relator_{\mathcal{D}} \relator_{\mathcal{E}} \relthree)$ 
  and let $u\ \relator_{\mathcal{E}} \relone\ v$. We prove 
  $f_*(u)\ \relator_{\mathcal{D}} \relator_{\mathcal{E}} \relthree\ g_*(v)$. 
  By very definition of $f_*$ we have two possible cases. 
  \begin{description}
    \item[Case 1.] 
      Suppose $f_*(u) = \delta(u)$ and $u = \inr(e)$, for some $e \in E$ 
      (we will omit subscripts in the Dirac's functions). Since $u\ \relator_{\mathcal{E}} \relone\ v$ 
      we have $v = \inr(e)$ as well, meaning that $g_*(v) = \delta(v)$. The thesis can now be rewritten as 
      $\delta(u)\ \relator_{\mathcal{D}} \relator_{\mathcal{E}} \relone\ \delta(v)$. 
      Since $\relator_{\mathcal{D}}$ satisfies property \eqref{Lax-Unit}, we have 
      $\relator_{\mathcal{E}} \relthree \subseteq (\delta_{\exception(X)} 
      \times \delta_{\exception(Y)})^{-1}(\relator_{\mathcal{D}} \relator_{\mathcal{E}} \relthree)$,
      meaning that the thesis follows from $u\ \relator_{\mathcal{E}} \relthree\ v$. 
      The latter indeed holds (since both $u = \inr(e) = v$),
      by very definition of $\relator_{\mathcal{E}}$.
    \item[Case 2.]
       Suppose now $f_*(u) = f(x)$ and $u = \inl(x)$, for some $x \in X$. The latter, together with 
       $u\ \relator_{\mathcal{E}} \relone\ v$ implies $v = \inl(y)$, for some $y \in Y$ such $x\ \relone\ y$. 
       In particular, $v = \inl(y)$ implies $g_*(v) = g(y)$. By hypothesis, 
       $x\ \relone\ y$ implies $f(x)\ \relator_{\mathcal{D}} \relator_{\mathcal{E}} \relthree\ g(y)$, i.e. 
       $f_*(v)\ \relator_{\mathcal{D}} \relator_{\mathcal{E}} \relthree\ g_*(v)$. 
  \end{description}
\end{proof}
  
\begin{lemma}
  The following holds
  $$\relone \subseteq (\eta_X \times \eta_Y)^{-1}(\relator_{\mathcal{D}} \relator_{\mathcal{E}} \relone).$$
\end{lemma}
\begin{proof}
  Suppose $x\ \relone\ y$ and recall that $\eta_Z = \delta_{\exception(Z)} \circ \epsilon_Z$. 
  Since $\relator_{\mathcal{E}}$ satisfies property \eqref{Lax-Unit} (wrt $\exception$), 
  from $x\ \relone\ y$, we infer $\epsilon_X(x)\ \relator_{\mathcal{E}} \relone\ \epsilon_Y(y)$. 
  We conclude that
  $\delta_{\exception(X)}(\epsilon_X(x))\ \relator_{\mathcal{D}} \relator_{\mathcal{E}} \relone\ \delta_{\exception(Y)}(\epsilon_Y(y))$
  holds,
  since $\relator_{\mathcal{D}}$ satisfies \eqref{Lax-Bind} as well (wrt $\distribution$).
\end{proof}

\begin{corollary}
  The map $\relator_{\mathcal{D}} \circ \relator_{\mathcal{E}}$ is a relator for the monad 
  $\distribution \circ \exception$.
\end{corollary} 

\begin{proof}
  By previous lemma it is sufficient to prove that
  given $\relone \subseteq (f \times g)^{-1}(\relator_{\mathcal{D}} \relator_{\mathcal{E}} \relthree)$
  we have 
  $\relator_{\mathcal{D}} \relator_{\mathcal{E}} \relone \subseteq 
  (\kleisli{f} \times \kleisli{g})^{-1}(\relator_{\mathcal{D}} \relator_{\mathcal{E}} \relthree)$,  
  i.e. $\relator_{\mathcal{D}} \relator_{\mathcal{E}} \relone  \subseteq (f_*^{\distribution} \times g_*^{\distribution})^{-1}(\relator_{\mathcal{D}} \relator_{\mathcal{E}} \relthree)$. Since $\relator_{\mathcal{D}}$ is a relator 
  for the monad $\distribution$, the latter is implied by 
  $\relator_{\mathcal{E}} \relone \subseteq (f_* \times g_*)^{-1}(\relator_{\mathcal{D}} \relator_{\mathcal{E}} \relthree)$, 
  which itself follows from $\relone \subseteq (f \times g)^{-1}(\relator_{\mathcal{D}} \relator_{\mathcal{E}} \relthree)$ 
  and Lemma \ref{laxProbEXcepHelp}.
\end{proof}

To conclude, we have to show that the monad $\distribution \circ \exception$
and the relator $\relator_{\mathcal{D}} \circ \relator_{\mathcal{E}}$ have the
required order-theoretic properties. First of all note that the monad
$\distribution \circ \exception$ carries a continuous $\signature$-algebra
structure, with the \ocppo\ order given by the functor $\distribution$.
Moreover, trivial calculations show that the monad is \ocppo-enriched (this
essentially follows from the order-enrichment of $\distribution$, together
with the validity of the equation $(\lub_{n<\omega} f_n)_* = \lub_{n<\omega}
f_{n_*}$), and thus the associated bind operator is continuous. Finally, it is
immediate to observe that $\relator_{\mathcal{D}} \circ
\relator_{\mathcal{E}}$ is inductive, since $\relator_{\mathcal{D}}$ is.

\section{Related Work}
As mentioned in the Introduction, this is certainly not the first
paper about program equivalence for higher-order effectful
calculi. Denotational semantics of calculi having this nature, has
been studied since Moggi's seminal work~\cite{Moggi/LICS/89}, thus
implicitly providing a notion of equivalence. All this has been given
a more operational flavour starting with Plotkin and Power account on
adequacy for algebraic effects~\cite{PlotkinPower/FOSSACS/01}, from
which the operational semantics presented in this paper is greatly
inspired. The literature also offers abstract accounts on logical
relations for effectful calculi. The first of them is due to
Goubault-Larrecq, Lasota and
Nowak~\cite{GoubaultLasotaNowak/MSCS/2008}, which is noticeably able
to deal with nondeterministic and probabilistic effects, but also with
dynamic name creation, for which applicative bisimilarity 
is known to be unsound.
Another piece of work which is related to ours is due to Johann, Simpson, and
Voigtl\"ander~\cite{JohannSimpsonVoigtlander/LICS/2010}, who
focused on algebraic effects and observational equivalence, and their
characterisation via CIU theorems and a form of logical relation
based $\top\top$-lifting. In both cases, the target language is
typed.
Similar in spirit to our approach (which is based on the notion 
of relator), the work of Katsumata and Sato \cite{Katsumata2013} analyses 
monadic lifting of relations in the context of $\top \top$-lifting.

Although no abstract account exists on applicative coinductive
techniques for calculi with algebraic effects, some work definitely
exists in some specific cases. As a noticeable example, the works by
Ong~\cite{Ong/LICS/1993} and Lassen~\cite{Lassen/PhDThesis} deal with
nondeterminism, and establish soundness in all relevant cases,
although full abstraction fails. The first author, together with
Alberti, Crubill\'e and
Sangiorgi~\cite{DalLagoSangiorgiAlberti/POPL/2014,CrubilleDalLago/ESOP/2014}
have studied the probabilistic case, where full abstraction can indeed
be obtained if call-by-value evaluation is employed.
\section{Conclusion}
This is the first abstract account on applicative bisimilarity for
calculi with effects. The main result is an abstract soundness theorem
for a notion of applicative similarity which can be naturally defined
as soon as a monad and an associated relator are given which on the
one hand serve to give an operational semantics to the algebraic operations,
and on the other need to satisfy some mild conditions in order for
similarity to be a precongruence. Soundness of bisimilarity is then
obtained as a corollary. Many concrete examples are shown to fit
into the introduced axiomatics. 
A notable example is the output monad, for which a definition of 
applicative similarity based on labeled transition systems as in 
e.g. \cite{CroleGordon/MSCS/99} is unsound, a fact that the authors 
discovered after noticing the anomaly, and not vice versa. 
Nevertheless, we defined a different notion of applicative similarity that 
fits into our framework and whose associated notion of bisimilarity 
(Definition \ref{gamma-bisimulation}) coincide with the usual notion of 
bisimilarity. 

A question that we have not addressed in this work, but which is quite
natural, is whether an abstract full-abstraction result could exist,
analogously to what, e.g., Johann, Simpson, and Voigtl\"ander obtained
for their notion of logical relation. This is a very
interesting topic for future work. It is however impossible to get
such a theorem without imposing some further, severe, constraints on
the class of effects (i.e. monads and relators) of interest, e.g.,
applicative bisimilarity is well-known not to be fully-abstract in
calculi with nondeterministic effects, which perfectly fit in the
picture we have drawn in this paper. A promising route towards this
challenge would be to understand which class of \emph{tests} (if any)
characterise applicative bisimilarity, depending on the underlying
monad and relator, this way generalising results by van Breugel,
Mislove, Ouaknine and
Worrell~\cite{vanBreugelMisloveOuaknineWorrell/TCS/2005} or
Ong~\cite{CrubilleDalLago/ESOP/2014}.

Finally, environmental bisimilarity is known \cite{KoutavasLevySumii/ENTCS/2011} 
to overcome the limits 
of applicative bisimilarity in presence of information hiding. 
Studying the applicability of the methodology developed in this  
work to environmental bisimilarity is yet another interesting
topic for future researches.
\subsection*{Acknowledgment}
The authors would like to thank Rapha\"elle Crubill\'e and 
the anonymous reviewers for the many useful comments, some of which 
led to a substantial improvement of our work. 
Special thanks go to Davide Sangiorgi, Ryo Tanaka, and
Valeria Vignudelli for many insightful discussions about the topics of
this work.

\bibliographystyle{plain}
\bibliography{mainBib}

\end{document}

%% file: macros.tex
\usepackage{amsmath}

\usepackage{mathrsfs}
\usepackage[utf8]{inputenc}
\usepackage[T1]{fontenc}
\usepackage{mathrsfs}
\usepackage{amsmath}
\usepackage{amsthm}
\usepackage{mdframed}

\usepackage{amssymb}
\usepackage{bussproofs}
\EnableBpAbbreviations

\usepackage{thrmappendix}
\usepackage{color}
\usepackage{url}
\usepackage{stmaryrd}
\usepackage{xypic}
\usepackage{scrextend}
\usepackage{mdwlist}

\newenvironment{varitemize}
{
\begin{list}{\labelitemi}
{\setlength{\itemsep}{0pt}
 \setlength{\topsep}{0pt}
 \setlength{\parsep}{0pt}
 \setlength{\partopsep}{0pt}
 \setlength{\leftmargin}{15pt}
 \setlength{\rightmargin}{0pt}
 \setlength{\itemindent}{0pt}
 \setlength{\labelsep}{5pt}
 \setlength{\labelwidth}{10pt}
}}
{
 \end{list}
}

\newcounter{numberone}

\newenvironment{varenumerate}
{
\begin{list}{\arabic{numberone}.}
{
  \usecounter{numberone}
  \setlength{\itemsep}{0pt}
  \setlength{\topsep}{0pt}
  \setlength{\parsep}{0pt}
  \setlength{\partopsep}{0pt}
  \setlength{\leftmargin}{15pt}
  \setlength{\rightmargin}{0pt}
  \setlength{\itemindent}{0pt}
  \setlength{\labelsep}{5pt}
  \setlength{\labelwidth}{15pt}
}}
{
\end{list}
}

\newcommand{\closedTerms}{\terms_0}
\newcommand{\closedValues}{\values_0}
\newcommand{\reloneTerms}{\relone_{\terms}}
\newcommand{\reloneValues}{\relone_{\values}}
\newcommand{\relthreeTerms}{\relthree_{\terms}}
\newcommand{\relthreeValues}{\relthree_{\values}}
\newcommand{\terms}{\ensuremath{\mathcal{T}}}
\newcommand{\values}{\ensuremath{\mathcal{V}}}

\newcommand{\subst}[3]{#1[#2/#3]}

\newcommand{\variables}{\mathtt{X}}

\newcommand{\valone} {V}
\newcommand{\varone} {x}
\newcommand{\vartwo} {y}

\newcommand{\termone}{M}
\newcommand{\termtwo}{N}
\newcommand{\termthree}{L}

\newcommand{\valtwo} {\mbox{\tt w}}
\newcommand{\valthree}{U}

\newcommand{\fv}[1]{FV(#1)}

\newcommand{\vars}{\tuple{\varone}}

\newcommand{\imp}{\vdash}

\newcommand{\raiseExcep}[1]{\mathsf{raise}_{#1}}
\newcommand{\exception}{\mathcal{E}}

\newcommand{\ctxone}{C}

\newcommand{\distone}{\mu}
\newcommand{\disttwo}{\nu}
\newcommand{\support}[1]{\mathsf{supp}(#1)}

\newcommand{\ruleBot}{(\mathsf{bot})}
\newcommand{\ruleReturn}{(\mathsf{ret})}
\newcommand{\ruleSeq}{(\mathsf{seq})}
\newcommand{\ruleApp}{(\mathsf{app})}
\newcommand{\ruleOp}{(\mathsf{op})}

\newcommand{\howeVar}{(\mathsf{How1})}
\newcommand{\howeAbs}{(\mathsf{How2})}
\newcommand{\howeRet}{(\mathsf{How3})}
\newcommand{\howeApp}{(\mathsf{How4})}
\newcommand{\howeSeq}{(\mathsf{How5})}
\newcommand{\howeOp}{(\mathsf{How6})}
\newcommand{\dependentProduct}[2]{\prod_{#1}#2}

\newcommand{\opT}{\op^{\monad}}

\newtheorem{lemma}{Lemma}
\newtheorem{definition}{Definition}
\newtheorem{proposition}{Proposition}
\newtheorem{theorem}{Theorem}
\newtheorem{example}{Example}
\newtheorem{remark}{Remark}
\newtheorem{corollary}{Corollary}

\newcommand{\signature}{\Sigma}

\newcommand{\abs}[1]{\lambda #1.}

\newcommand{\ocpo}{\ensuremath{\omega\mathbf{CPO}}}

\newcommand{\ocppo}{\ensuremath{\omega\mathbf{CPPO}}}

\newcommand{\powersetmonad}{\mathcal{P}}

\newcommand{\lan}{\langle}
\newcommand{\ran}{\rangle}

\newcommand{\hh}{\hdots}

\newcommand{\HomSet}[3]{\mathsf{Hom}_{#1}(#2,#3)}

\newcommand{\powerset}{\mathcal{P}}
\newcommand{\distribution}{\mathcal{D}}

\newcommand{\inl}{\mathsf{in}_l}
\newcommand{\inr}{\mathsf{in}_r}
\newcommand{\monad}{T}
\newcommand{\unit}[1]{\eta(#1)}
\newcommand{\uniT}{\eta}

\newcommand{\relthree}{\mathcal{S}}

\newcommand{\functor}{F}

\newcommand{\ctxequiv}{\equiv}

\newcommand{\rel}{\mathcal{R}}
\newcommand{\howe}[1]{#1^{H}}
\newcommand{\relone}{\mathcal{R}}

\newcommand{\bisim}{\sim}
\newcommand{\gfp}{\mathsf{gfp}}

\newcommand{\cpoone}{D}
\newcommand{\cpotwo}{E}

\newcommand{\Cpoone}{\mathbb{D}}
\newcommand{\Cpotwo}{\mathbb{E}}
\newcommand{\Cpothree}{\mathbb{C}}
\newcommand{\cpoleq}{\sqsubseteq}
\newcommand{\lub}{\bigsqcup}

\newcommand{\ocppoCat}{\omega\mathbb{C}\mathbb{P}\mathbb{P}\mathbb{O}}

\usepackage{enumitem, hyperref}
\makeatletter
\def\namedlabel#1#2{\begingroup
    #2%
    \def\@currentlabel{#2}%
    \phantomsection\label{#1}\endgroup
}
\makeatother

\newcommand{\bps}{\begin{description}}
\newcommand{\eps}{\end{description}}

\newcommand{\similar}{\precsim}
\newcommand{\howesimilar}{\similar^{H}}
\newcommand{\cosimilar}{\succsim}

\newcommand{\dist}{\mathcal{D}}

\newcommand{\bind}{\texttt{>>=}}

\newcommand{\sem}[1]{\llbracket #1 \rrbracket}

\newcommand{\kleisli}[1]{#1^{\dagger}}
\newcommand{\return}{\mathsf{return}}

\renewcommand{\terms}{\Lambda}

\newcommand{\fun}[2]{#1 \mapsto #2}
\renewcommand{\valone}{V}
\renewcommand{\valtwo}{W}

\renewcommand{\termone}{M}
\renewcommand{\termtwo}{N}

\newcommand{\op}{\sigma}
\newcommand{\functional}{\mathcal{F}}

\newcommand{\similarTerms}{\similar_{\terms}}
\newcommand{\similarValues}{\similar_{\values}}

\newcommand{\compone}{X}
\newcommand{\comptwo}{Y}

\newcommand{\catone}{\mathbb{C}}
\newcommand{\TO}[1]{\Rightarrow^{n}}

\newcommand{\evalTo}[3]{#1 \Downarrow_{#2} #3}	

\newcommand{\lfp}[1]{\mu #1}
\renewcommand{\gfp}[1]{\nu #1}
\newcommand{\relunit}{\mathcal{U}}
\newcommand{\ctxpreord}{\leq}
\newcommand{\behavior}{\mathcal{B}}
\newcommand{\open}[1]{#1^{\circ}}
\newcommand{\tuple}[1]{\bar{#1}}
\newcommand{\ctxclosure}[1]{\widehat{#1}}

\newcommand{\vals}{\tuple{\valone}}

\newcommand{\evalToo}[2]{\evalTo{#1}{}{#2}}

\newcommand{\seq}[2]{#1\ \mathsf{to}\ x.#2}
\newcommand{\substcomp}[3]{#1[#2 := #3]}
\newcommand{\stream}[1]{#1^{\infty}}
\newcommand{\concat}[2]{#1::#2}
\newcommand{\print}[2]{\mathsf{print}\ #1. #2}
\renewcommand{\return}[1]{\mathsf{return}\ #1}
\newcommand{\relator}{\Gamma}
\renewcommand{\rel}[2]{2^{#1 \times #2}}
\newcommand{\relatorSim}{\Gamma}
\newcommand{\relatorBisim}{\Delta}
\newcommand{\conversive}[1]{#1^{c}}

\newcommand{\set}{\mathbb{S}\mathbb{E}\mathbb{T}}

\newcommand{\kleisliCat}[1]{\mathcal{K}\ell(#1)}
\newcommand{\coalgX}{\gamma}

\newcommand{\contFun}[2]{#1 \xrightarrow{c} #2}